\documentclass[letterpaper, 11 pt]{article}

\usepackage{amsmath, amssymb, amsthm, amsfonts}
\usepackage{authblk}
\usepackage{graphicx}
\usepackage{times}
\usepackage{enumerate}
\usepackage{xspace}
\usepackage{dirtytalk}
\usepackage{setspace}
\usepackage[margin = 1 in]{geometry}
\usepackage[parfill]{parskip}
\usepackage{soul}
\usepackage{comment}

\usepackage{algorithm}
\usepackage[noEnd=true]{algpseudocodex}
\usepackage{todonotes}
\usepackage{tabularray}
\usepackage{tabularx}
\MakeRobust{\Call}
\usepackage{makecell}
\usepackage{xcolor}

\usepackage{wrapfig}

\usepackage{xcolor}

\usepackage[makeroom]{cancel}
\usepackage{mathtools}
\usepackage{array}
\usepackage{url, hyperref}
\newcolumntype{P}[1]{>{\centering\arraybackslash}p{#1}}

\newtheorem{theorem}{Theorem}[section]
\theoremstyle{definition}

\newtheorem{lemma}[theorem]{Lemma}
\newtheorem{defn}[theorem]{Definition}

\newtheorem{claim}[theorem]{Claim}
\newtheorem{corollary}[theorem]{Corollary}

\newtheorem{observation}[theorem]{Observation}

\theoremstyle{remark}
\newtheorem{remark}{Remark}

\makeatletter
\algdef{SE}[UPON]{Upon}{EndUpon}[1]{\algpx@startIndent\algpx@startCodeCommand\textbf{Upon}\ #1\ \algorithmicdo}{\algpx@endIndent\algpx@startCodeCommand\algorithmicend\ Upon}%

\ifbool{algpx@noEnd}{%
  \algtext*{EndUpon}%
  \apptocmd{\EndUpon}{\algpx@endIndent}{}{}%
}{}%

\pretocmd{\Upon}{\algpx@endCodeCommand}{}{}

\ifbool{algpx@noEnd}{%
  \pretocmd{\EndUpon}{\algpx@endCodeCommand[1]}{}{}%
}{%
  \pretocmd{\EndUpon}{\algpx@endCodeCommand[0]}{}{}%
}%
\makeatother

\newcommand{\invec}{\textbf{X}}

\newcommand{\invX}[1]{#1^{\textsf{inv}}}
\newcommand{\invB}{\byz^{\textsf{inv}}}

\newcommand{\randvec}{\textbf{R}}

\newcommand{\good}{\textsf{Good}}
\newcommand{\byz}{\textsf{Byz}}

\newcommand{\goodv}{\mathcal{M}_\good(v)}
\newcommand{\goodpv}{\mathcal{M}'_\good(v)}

\newcommand{\byzv}{\mathcal{M}_\byz(v)}
\newcommand{\byzpv}{\mathcal{M}'_\byz(v)}

\newcommand{\allv}{\mathcal{M}(v)}
\newcommand{\allpv}{\mathcal{M}'(v)}

\def\ConstDecTree{\mbox{\sc Const\_Decision\_Tree}}
\def\Determine{\mbox{\sc Determine}}

\def\inline#1:{\par\vskip 7pt\noindent{\bf #1:}\hskip 10pt}
\def\dnsinline#1:{\par\vskip -7pt\noindent{\bf #1:}\hskip 10pt}

\def\inline#1:{\par\vskip 7pt\noindent{\bf #1:}\hskip 10pt}
\def\dnsinline#1:{\par\vskip -7pt\noindent{\bf #1:}\hskip 10pt}

\def\EXSET{\mathcal{E}}
\def\Source{source\xspace}

\def\peers{peers\xspace}

\def\peer{peer\xspace}
\def\nonfaulty{nonfaulty\xspace}
\def\Advers{\mathcal{A}}

\def\Inum{\mathcal{K}}

\newcommand{\modof}[1]{\lvert#1\rvert}

\def\indID{ind}
\def\indID{\ell}
\def\Exp{\hbox{\rm I\kern-2pt E}}

\long\def\commentstart #1\commentend{}

\def\byzfrac{\beta}
\def\goodfrac{\gamma}

\def\lead{\mathtt{lead}}
\def\minbit{I_{min}}

\def\peers{peers\xspace}

\def\peer{peer\xspace}
\def\BYZ{\mathcal{F}} 
\def\view{\mbox{\sf view\xspace}\xspace}
\def\viewchange{\mbox{\sf view change\xspace}\xspace}
\def\indx{I\xspace}
\def\res{res}

\def\ith{$i^{\text{th}}$\ }
\def\jth{$j^{\text{th}}$\ }

\def\DataQuery{\mbox{\tt Query}}
\def\Time{\mathcal{T}}
\def\Query{\mathcal{Q}}

\def\download{{\textsf{Download}}}

\def\phase{{\mathtt{phase}}}
\def\stage{{\mathtt{stage}}}

\def\metoo{``me neither"}
\def\EXEC{\mathtt{EX}}

\newcommand{\BB}{\mathcal{B}}
\newcommand{\RR}{\mathcal{R}}
\newcommand{\hRR}{{\hat{\mathcal{R}}}}
\def\cH{\mathcal{H}}

\def\FS{\mbox{\sf FS}}

\def\Time{\mathcal{T}}
\def\optqexact{O\left(\frac{n\lg{n}}{\goodfrac k}\right)}
\def\optqtilde{\tilde{O}\left(\frac{n}{\goodfrac k}\right)}
\def\Query{\mathcal{Q}}

\def\Message{\mathcal{M}}
\def\Messagesize{\mathcal{S}}

\def\leftchild{\mbox{\sf\small left-child}}
\def\rightchild{\mbox{\sf\small right-child}}

\title{Distributed Download from an External Data Source in \\
Faulty Majority Settings}

\author[1]{John Augustine \thanks{\href{mailto: augustine@iitm.ac.in}{\tt augustine@iitm.ac.in}}}
\author[1]{Soumyottam Chatterjee \thanks{\href{mailto: soumyottam@acm.org}{\tt soumyottam@acm.org}}}
\author[2]{Valerie King \thanks{\href{mailto: val@uvic.ca}{\tt val@uvic.ca}}} 
\author[1]{Manish Kumar \thanks{\href{mailto: manishsky27@gmail.com}{\tt manishsky27@gmail.com}}}
\author[3]{Shachar Meir \thanks{\href{mailto: shachar.meir@weizmann.ac.il}{\tt shachar.meir@weizmann.ac.il}}}
\author[3]{David Peleg \thanks{\href{mailto: david.peleg@weizmann.ac.il}{\tt david.peleg@weizmann.ac.il}}}

\affil[1]{Indian Institute of Technology Madras, India}
\affil[2]{University of Victoria, Canada}
\affil[3]{Weizmann Institute of Science, Israel}

\begin{document}

\maketitle

\begin{abstract}

We extend the study of retrieval problems in distributed networks, focusing on improving the efficiency and resilience of protocols in the \emph{Data Retrieval (DR) Model}. The DR Model consists of a complete network (i.e., a clique) with $k$ peers, up to $\beta k$ of which may be Byzantine (for $\beta \in [0, 1)$), and a trusted \emph{External Data Source} comprising an array $X$ of $n$ bits ($n \gg k$) that the peers can query. Additionally, the peers can also send messages to each other. In this work, we focus on the \download\ problem that requires all peers to learn $X$. Our primary goal is to minimize the maximum number of queries made by any honest peer and additionally optimize time.

We begin with a randomized algorithm for the \download\ problem that achieves optimal query complexity up to a logarithmic factor. For the stronger dynamic adversary that can change the set of Byzantine peers from one round to the next, we achieve the optimal time complexity in peer-to-peer communication but with larger messages. In broadcast communication where all peers (including Byzantine peers) are required to send the same message to all peers, with larger messages, we achieve almost optimal time and query complexities for a dynamic adversary. Finally, in a more relaxed crash fault model, where peers stop responding after crashing, we address the \download\ problem in both synchronous and asynchronous settings. Using a deterministic protocol, we obtain nearly optimal results for both query complexity and message sizes in these scenarios.

\end{abstract}

\section{Introduction} \label{sec: introduction}

We consider the Data Retrieval Model (DR) which is comprised of a peer-to-peer network and an external data source in the form of an $n$ bit array. There are $k$ peers, some of which may be faulty. The peers initially do not know the array's contents but must learn all its bits. The peers can learn the bits either through direct querying of the data source or from other peers. The primary goal is to minimize the maximum number of queries by any peer.  

The DR model was introduced in DISC 2024~\cite{ABMPRT24} and was inspired by distributed oracle networks (DONs) which are a part of blockchain systems. In such networks, nodes are tasked with retrieving information from external data sources such as stock prices. We believe that the problem is of more general interest. Consider a collection of facts about the real world, each of which may be learned by deep investigation. This work may be shared among researchers to reduce the cost for any individual, even if some fraction of the researchers may be Byzantine or unreliable.

Here we focus on the fundamental \download\ problem, where each peer must learn every bit in the array. The problem is easily solved in a query-balanced manner in the absence of failures. When faults are allowed, this becomes harder. We consider a setting where up to $\byzfrac k$ are faulty and at least $\goodfrac  \geq  1 - \byzfrac$ fraction are non-faulty.

For synchronous systems, a tight bound on query complexity is established in~\cite{ABMPRT24} for deterministic \download, complemented by two randomized protocols that solve \download\ w.h.p. The first can tolerate any fraction $\byzfrac <1$ of Byzantine faults but has non-optimal query complexity of $O(n/\goodfrac k + \sqrt{n})$, while the second has optimal query complexity of $\tilde{O}(n/\goodfrac k)$\footnote{We use the $\tilde{O}(\cdot)$ notation to hide $\byzfrac$ factors and polylogarithms in $n$ and $k$.} but can only tolerate up to $\byzfrac <1/3$ fraction of Byzantine faults. In this paper, we first present a novel randomized protocol in a synchronous network that combines the best of both results, achieving optimal query complexity while tolerating any fraction $\byzfrac < 1$ of Byzantine faults.

We then proceed to provide a lower bound on time/query complexity. Specifically, we show that in any {\em single round} randomized protocol, every peer must essentially query the entire input. On the other hand, we show that $\tilde{O}(1)$ rounds suffice to bring down {\it expected} query cost to optimal, at the cost of large message size. By adding the assumption of a Broadcast model where all peers including faulty ones are restricted to sending the same message to all peers, we are able to bring down the message size and simultaneously achieve optimal query size, time, and message size, to within log factors. Furthermore, unlike the protocols in \cite{ABMPRT24}, these last two protocols allow for a ``dynamic" adversary where the set of faulty \peers can be arbitrarily changed in each round, provided that in any given round, at most $\byzfrac$ fraction of peers are faulty.

Lastly, we turn to the easier setting of crash failures, for which we are able to provide query-efficient deterministic algorithms in both synchronous and {\em asynchronous} networks.

\subsection{The Model}  The Data Retrieval model consists of (i) $k$ \peers\ that form a \emph{clique} and (ii) a source of data that is external to the clique called the \emph{\Source} that stores the input array comprising $n$ bits and provides read-only access to its content through queries.

\subparagraph*{Clique network and communication mode.}  In a clique (complete) network, the $k$ \peers\ are identified by unique IDs assumed to be from the range $[1,k]$.

In each round, every \peer can send a message of up to $b$ bits to each of the other peers. The common variant of this communication mechanism, referred to as \emph{peer-to-peer message passing} communication, allows a peer to send in each round a \emph{different} message to each of the other peers. However, we also discuss (in Sect. {\ref{ss:opt_in_broadcast_model}}) a variant termed \emph{broadcast} communication, where each peer (including a faulty one) can send at most one message per round, and that message is delivered to all other peers.

\subparagraph*{The \Source.} 
The $n$-bit input array\footnote{Throughout this paper we assume $n\ge k$. In typical applications, $n \gg k$.} $\invec=\{b_1,\ldots,b_n\}$ is stored in the \Source. It allows \peers\ to retrieve that data through queries of the form $\DataQuery(i)$, for $1\le i\le n$. 
The answer returned by the \Source\ would then be $b_i$, the $i^{th}$ element in the array.
This type of communication is referred to as \emph{\Source-to-\peer} communication.

\subparagraph*{Network delay and rounds.}
We consider both synchronous and asynchronous settings. In the synchronous setting, \peers\ share a global clock.
Each round consists of three sub-rounds:
\begin{enumerate}
    \item ~The \emph{query sending sub-round}, in which each peer can send $q$ queries ($0\le q\le n$) of the form $\DataQuery(\cdot)$ to the \Source.
    
    \item The \emph{query response sub-round}, in which the \Source responds to all the queries.
    
    \item The \emph{message-passing sub-round} of \peer-\peer\ communication, consisting of messages exchanged between peers. Every message is of size $O(\log n)$ unless otherwise stated.
\end{enumerate}
We assume that local computation is instantaneous and is performed at the beginning of each sub-round.
We assume that a \peer\ $M$ can choose to ignore (not process) messages received from another peer during the execution.  Such messages incur no communication cost to the recipient peer. 

In section \ref{sec:async_crash_faults} we consider a fully asynchronous communication where an adversary may delay every message by any finite amount of time. 

\subparagraph*{The adversarial settings.}
The behavior of the environment in which our protocols operate is modeled by an adversary $\Advers$ that selects the input data and determines the \peers' failure pattern. In executing a protocol, a \peer is considered \emph{\nonfaulty} if it obeys the protocol throughout the execution. 

In this work, we consider two types of \emph{faulty} \peers. A \emph{crashed} \peer is one that stops its local execution of the protocol arbitrarily and permanently (controlled by $\Advers$). This might happen \emph{in the middle of a sub-round}, that is, after the \peer has already sent some but not all of its messages.
A \emph{Byzantine} \peer\
can deviate from the protocol arbitrarily
(controlled by $\Advers$) .
The adversary $\Advers$ can corrupt at most $\byzfrac k$ \peers\ for some given\footnote{We do not assume~$\byzfrac$ to be a fixed constant (unless mentioned otherwise).}
$\byzfrac \in [0,1)$. 
Letting $\goodfrac = 1-\byzfrac$,
there is (at least) a $\goodfrac$ fraction of \nonfaulty \peers. 
We denote the set of faulty (respectively, \nonfaulty) \peers\ in the execution by $\BYZ$. (resp., $\cH$). 

In both cases the total number of corrupted peers at any given time is bounded by $\byzfrac k$. In both cases, for a randomized algorithm, the adversary can adaptively decide on which peers to corrupt.

The Byzantine adversary can select peers to corrupt at the start of each round. Corrupted peers are called Byzantine peers and they can behave arbitrarily. We consider two types of Byzantine adversaries. Under the \emph{fixed adversary}, a corrupted peer remains corrupted for the rest of the execution. 
The \emph{dynamic} adversary can
decide on the set of corrupt peers arbitrarily at the start of any round, or more explicitly, it can make a peer $v$ faulty on one round and non-faulty on the next.
In both cases the total number of corrupted peers at any given time is bounded by $\byzfrac k$. In both cases, for a randomized algorithm, the adversary can adaptively decide on which peers to corrupt.

When assuming a synchronous network, the adversary has no control over the delay of the network, but when assuming an asynchronous network, the adversary $\Advers$ can delay messages from arriving at their destinations for any finite amount of time.

We design both deterministic and randomized protocols. When the protocol is deterministic, the adversary can be thought of as all-knowing. Thus, $\Advers$ 
knows exactly how the complete execution 
will proceed and can select Byzantine \peers from the beginning based on this knowledge. When the protocol is randomized,  the \peers\ may generate random bits locally.
At the beginning of each round~$t$, $\Advers$ has knowledge of $\invec$, all the local random bits generated up to round $t-1$, and all \peer-\peer\ and \Source-\peer\ communications up to round $t-1$.

\subparagraph*{Complexity measures.}
The following complexity measures are used to analyze our protocols. 
\begin{description}
\item{(i)} 
\emph{Query} Complexity ($\Query$): the maximum number of bits queried by a \nonfaulty \peer\ during the execution of the protocol,
\item{(ii)}
\emph{Round} Complexity ($\Time$): the number of rounds (or \emph{time}) it takes for the protocol to terminate,
\item{(iii)}
\emph{Message} Complexity ($\Message$): the total number of messages sent by \nonfaulty \peers\ during the execution of the protocol.
\item{(iv)}
\emph{Message size}($\Messagesize$)
: the maximum number of bits sent in one message by any \nonfaulty \peer\ during the execution of the protocol. 
\end{description}

As queries to the \Source\ are expected to be the more expensive component in the foreseeable future, we focus mainly on optimizing the query complexity $\Query$. 
Note that our definition of~$\Query$ (measuring the maximum cost per \peer\ rather than the total cost) favors a fair and balanced load of queries across \nonfaulty \peers.

\paragraph*{The \download\ problem and its complexity in a failure-free setting.}  A natural class of problems, called \emph{retrieval problems}, arises from the definition of the DR model. In a retrieval problem, each \peer needs to output some computable boolean function $f$ of the \Source (i.e $f(\invec)$). In this work, we focus on the case where $f$ is the identity function and hereafter refer to this particular problem as the \download\ problem. This problem is the most fundamental retrieval problem since every computable function $f$ of the input can be computed by the \peers by first running a \download\ protocol and then computing $f(\invec)$ locally at no additional cost. Hence, query cost of the \download\ problem serves as a baseline against which to compare the costs of other specialized algorithms for specific problems. 
Observe that a $\Query$ lower bound for computing any Boolean function on $\invec$ serves as a lower bound for \download\ as well.

\paragraph{The \download\ problem.}  Consider a DR network with $k$ \peers, where at most $\byzfrac k$ can be faulty, and a \Source that stores an array of bits $\invec=[b_1, \dots, b_n]$. The \download\ problem requires each \peer $M$ to learn $\invec$. Formally, each non-faulty peer $M$ outputs a bit array $\res_M$, and it is required that, upon termination, $\res_M[i]=b_i$ for every $i \in \{1, \cdots, n\}$ and $M \in \cH$, where $\cH$ is the set of non-faulty peers.

To solve this problem in the absence of failures, all $n$ bits need to be queried, and this workload can be shared evenly among $k$ \peers, giving $\Query = \Theta(n/k)$. The message complexity is $\Message = \tilde{O}(n k)$, assuming small messages of size $\tilde{O}(1)$, and round complexity is $\Time = \tilde{O}(n/k)$ since  $\Omega(n/k)$ bits need to be sent along each communication link when the workload is shared. 

\subsection{Methods} \label{sec:methods}

In this subsection, we highlight the main tools used throughout this work.

\subparagraph*{Blacklisting.} 
During an execution, \nonfaulty \peers\ can \emph{blacklist} faulty ones, after identifying a deviation from the behavior expected of a \nonfaulty \peer, and subsequently ignore their messages. A Byzantine
\peer $M'$ can be blacklisted for a variety of reasons, such as
when $M'$ is directly ``caught'' in a lie about the value of some bit, or stops sending messages while they are expected of it, or sends more messages than what they are expected to send.

\subparagraph*{A primary-backup approach.}  In some of our protocols, we use the \emph{primary-backup} approach, first introduced in {\cite{primay-backup}}. In this approach, {\peers} move through a succession of configurations called \emph{views}, during which one {\peer} will be designated the \emph{leader} (or \emph{primary}), in charge of driving progress and coordinating the rest of the {\peers} (or \emph{backups}). In case the leader is faulty (which can be recognized in several ways, including but not limited to the blacklisting techniques highlighted above), {\peers} can initiate a \emph{view change} which will result in a new leader being selected.

\subparagraph*{Sifting and the use of decision-trees.}  This technique is used in Section {\ref{sec:alg_fast}. Consider a multi-set of the strings proposed by different peers that purport their respective strings to be equal to a particular interval of the input bit array. If it is known that at least $t>0$ proposers were nonfaulty peers that correctly know the interval, we can discard all strings in the multi-set that do not appear at least $t$ times. We can conclude that the remaining set of distinct strings contains a correct string.

A {\em decision tree} for a set $S$  of strings is a rooted binary tree.
Each internal node $x$ is labeled by an index $i$ of the input array and each leaf is labeled by a string $s$ such that if a path from the root to the leaf goes to the left subtree of a node labeled $i$, then the $i^{th}$ bit of $s$ is 0, else it's 1. Given a set of strings $S$ of which one is consistent with the input array, we can build a decision tree with $|S|-1$ nodes and determine the correct leaf with $|S|-1$ simultaneous queries.
 
\subsection{Related Work} 

Our work studies a new class of fault tolerant problems that is heavily inspired by Blockchain oracles.
There are multiple classic BFT problems (e.g., agreement, broadcast, and state machine replication) that provide insight and inspiration when considering the \download\ problem. 
 
\subparagraph*{Traditional BFT problems.}
The theory of Byzantine fault tolerance has been a fundamental part of distributed computing  ever since its introduction by Pease, Shostak, and Lamport~\cite{LSM82, PSL80} in the early 80's,
and has had a profound influence on cryptocurrencies, blockchains, distributed ledgers, and other decentralized P2P systems. 
It largely focused on a canonical set of  problems like Broadcast~\cite{DS83},  Agreement~\cite{B87,LSM82,PSL80,R83}, and State Machine Replication~\cite{CL99}. In most of these studies, the main parameter of interest is the maximum fraction $\byzfrac$ of the \peers\ that can be corrupted by the adversary in an execution. 

Consider the Byzantine Agreement problem that requires $n$ \peers, each with an input bit, to agree on a common output bit that is \emph{valid}, in the sense that at least one \nonfaulty (non-Byzantine) \peer\ held it as input. In the synchronous setting, even without cryptographic assumptions, there are agreement algorithms that can tolerate any fraction $\byzfrac < 1/3$ of Byzantine \peers~\cite{LSM82} (and this extends to asynchronous settings as well~\cite{B87}). When $\byzfrac \ge 1/3$, agreement becomes impossible in these settings~\cite{LSM82}. However, the bound improves to $\byzfrac < 1/2$ with message authentication by using cryptographic digital signatures~\cite{RSA78}.  By the well-known network partitioning argument (discussed shortly), $\byzfrac < 1/2$ is required for any form of Byzantine agreement. For most of the Byzantine fault tolerance literature, $\byzfrac$ hovers around either 1/3 or 1/2, with some notable exceptions like authenticated broadcast~\cite{DS83} that can tolerate any $\byzfrac < 1$. 

The main reason for this limitation stems from the inherent coupling of data and computing. Consider, for instance, any Byzantine Agreement variation with $\byzfrac \ge 1/2$.  When all \nonfaulty \peers\ have the same input bit (say, 1), the Byzantine \peers\ hold at least half the input bits and can unanimously claim 0 as their input bits. This ability of Byzantine \peers\ to spoof input bits makes it fundamentally impossible for \nonfaulty \peers\ to reach a correct agreement with the validity requirement intact. At the heart of this impossibility is the 
adversary's power to control information
crucial to solving the problem. 
In fact, this issue leads to many impossibilities and the inability to solve problems exactly (see e.g., \cite{AMP21}). 

In contrast, having a reliable \Source\ that provides the data in read-only fashion yields a distributed computing context where access to data cannot be controlled by Byzantine \peers. Taken to the extreme, any \nonfaulty \peer\ can individually solve all problems by directly querying the \Source\ for all required data. However, queries are charged for and can be quite expensive. So the challenge is to design effective and secure collaborative techniques to solve the problem at hand while minimizing the number of queries made by each \nonfaulty \peer\footnote{Note that appointing some individual \peers\ to query each input bit and applying a \emph{Byzantine Reliable Broadcast (BRB)} protocol \cite{ECCB,B87,DS83} for disseminating the bits to all \peers\ will not do, since the appointed \peers\ might be Byzantine, in which case the BRB protocol can only guarantee agreement on \emph{some} value, but not necessarily the true one. Moreover, Byzantine Reliable Broadcast (BRB) cannot be solved when $\byzfrac \geq 1/3$ with no authenticated messages.}.
Hence, despite the \Source\ being passive (read-only with no computational power), its reliability makes the model stronger than the common Byzantine model.

One problem that stands out among canonical BFT problems is the Byzantine Reliable Broadcast (BRB) problem, first introduced by Bracha \cite{B87}. In BRB, a designated sender holds a message $M$, and the goal is for every \nonfaulty \peer to output the same $M'$, which in addition must satisfy $M' = M$ if the sender is \nonfaulty. The \download\ problem can be viewed as a variant of BRB, where the sender is always \nonfaulty but has no computational powers and is passive (read-only), and \peers are always required to output the correct message $M$. 
These differences distinguish \download\ from BRB and require us to shift our focus to minimizing queries (while keeping the upper bound $\byzfrac$ on the fraction of faults as high as possible). One easy-to-see difference in results is that \download\ can be solved trivially even when $1/3 \leq \byzfrac < 1$ and there are no authenticated messages, whereas BRB can not be solved under the same conditions \cite{DS83}. Another difference is that state-of-the-art BRB protocols like \cite{ECCB} where the sender uses error-correcting codes and collision-resistant hash functions are inapplicable (when considering the \Source\ to be the sender).
In optimal \emph{balanced} BRB protocols like in \cite{ECCB}, the sender sends $O(\frac{n}{k})$ bits to each \peer whereas it is shown in \cite{ABMPRT24} that deterministic \download\ requires $\Omega(\byzfrac n)$ queries (the difference stems from the inability of the \Source\ to perform computations).

\paragraph{Oracle networks.}  As mentioned above, in Oracle networks a set of \peers are assigned the task of bringing external off-chain data to the network, where a subset of these \peers can be faulty. Generally, one can describe the operation of an Oracle network as follows. The Oracle network generates a report containing the observations made by some (sufficient) number of \peers.
Once a report is successfully generated, one or multiple \peers transmit the report to an intermediary program that executes on the blockchain (known as a smart contract) that verifies the validity of the report, derives a final value from it (e.g., the median), and then exposes the value for consumption on-chain.
Since the traditional usage of these networks is to track exchange rates (e.g USD-ETH) that change with time, studies on Oracle networks focus on the problem of creating a report where the derived value (say the median) must be acceptable (in the sense that it does not deviate much from the set of \nonfaulty observations) while keeping costs (sending reports to the smart contract has a relatively high cost) low and tolerating as many Byzantine \peers as possible, even at the expense of higher communication and computation off-chain.
The Off-Chain Reporting(OCR) protocol \cite{ocr} solves this problem with $\byzfrac<1/3$ by running a BA protocol to agree on $2f+1$ values, then a designated \emph{leader} generates a report and sends it to the contract (the leader acts as an aggregator).
The Distributed ORacle Agreement (DORA) protocol \cite{DORA} takes it a step further by using an approximate agreement scheme, using the inherent ability of a Blockchain to act as an ordering service, and multiple aggregators. They improve results to sustain $\byzfrac<1/2$ and $\byzfrac<1/3$ w.h.p when the size of the Oracle network is significantly smaller than the size of the entire system. 
In both of these works, every \peer reads all the external data and goes on to participate in the report generation.
Our work complements the OCR and DORA protocols and focuses on how to efficiently read the off-chain data while minimizing the number of bits read per \peer (as reading from an external source is also more costly than off-chain communication). Our approach would drastically reduce cost when the Oracle network keeps track of a large number of (static) variables (e.g., financial information) and could be used as a black box in the OCR and DORA protocols, every \peer would query all of the $n$ values individually, whereas in our solution that is not the case. Note that both OCR and DORA use cryptographic primitives, whereas we do not.

\subsection{Our Contributions}

We explore the \download\ problem under various adversarial and network models and present several deterministic and randomized protocols and a lower bound for \download. Here, we state only simplified bounds, in which the $\tilde{O}(\cdot)$ notation hides factors dependent on~$\byzfrac$ and poly log factors in $n$. The main results are summarized in Tables \ref{tbl: comparative_analysis} and \ref{tbl: sync_results} for convenience.

\subsubsection{Query Optimality with Synchronous Point-to-point Communication and Byzantine Failures}

We start with closing the gap left open in \cite{ABMPRT24}. 
The model studied in that paper involves a synchronous point to point communication network and Byzantine failures. In this model, for deterministic algorithms, the \download\ problem turns out to be expensive, requiring $\Omega(\byzfrac n)$ queries in the worst case (a matching upper bound is shown which works under asynchrony as well). 
Every \peer essentially has to query the entire input array for itself. However, \cite{ABMPRT24} gives a randomized algorithm that solves the \download\ problem (and consequently \emph{any} function of the input) for an \emph{arbitrary} fraction $\byzfrac < 1$ of Byzantine faults while requiring at most $\tilde{O}(n/k + \sqrt{n})$ queries per \peer. The result is nearly as efficient as the \emph{failure-free model} whenever $k < \sqrt{n}$. 
The time and message costs are $\Time= O(n)$ and 
$\Message= \tilde{O}(kn + k^2\sqrt{n})$.
A natural question then, is whether the additive $\sqrt{n}$ term is necessary for $k > \sqrt{n}$. It was shown in \cite{ABMPRT24} that as long as $\byzfrac < 1/3$, one can be fully efficient for all $k \in [1, n]$,
getting $\Query = \tilde{O}\left(\frac{n}{k}\right)$, 
$\Time = \tilde{O}(n)$, and
$\Message = \tilde{O}(nk^2)$.
In Section \ref{sec: optimal_time_P2P} we close the gap and show the existence of a randomized \download\ algorithm with query complexity $\Query = O\left(\frac{n}{\gamma k}\right)$, $\Time = O(n \log k)$, and
$\Message = O(nk^2)$ for $\beta \in [0,1)$.

\subsubsection{Faster Query-optimal Solutions in Synchronous Communication and Byzantine Failures}

We next ask whether \download\ can be achieved faster than linear time. To exclude the extreme end of the scale, we show (in Sect. \ref{section-lower-bounds}) that hoping for a single round \download\ (assuming arbitrarily large messages can be sent in a single round) is too ambitious. 
The only way to achieve this is via the trivial exhaustive algorithm, where every \peer queries every bit. More explicitly, allowing each \peer to query $(n - 1)$ bits is not enough to solve the \download\ problem. (for $\byzfrac \approx 1/2$).

Nevertheless, we next derive faster algorithms (in section \ref{sec:alg_fast}) than the one of Section \ref{sec: optimal_time_P2P}. Specifically, in Subsect. \ref{sec: 2-step algorithm} we show how \download\ can be achieved in \emph{two} rounds. 
This solution enjoys the additional advantage that it can cope with a stronger type of adversary, termed \emph{dynamic adversary}, which can change the set of Byzantine \peers from one round to the other, provided that this set never exceeds a size of $\byzfrac k$.

Unfortunately, this algorithm no longer attains query-optimality; its query complexity is $O(n/(\gamma k) + \sqrt{n})$. We improve this query complexity in Subsection \ref{ss:expected} where we describe an iterated version of the 2-step algorithm with {\it expected} query complexity $O(n\log n/(\gamma k))$ and $O(\log n)$ time.

Finally, in Subsection \ref{ss:opt_in_broadcast_model}
we show how the same nearly optimal results can be achieved, with worst case query complexity and small message message size, provided we assume {\em broadcast} communication among peers rather than point-to-point. That is, in each round, each \peer must send the same message to all other \peers, and this applies also to the Byzantine \peers.
In this model, we get an algorithm with worst case
query complexity $O((1/\gamma)\log^2 n)$ and $O((1/\gamma)\log^2 n)$ time, and message size $O(\log n/\gamma)$

\subsubsection{Crash Faults} 
We then turn to studying settings that allow only crash failures. Here, we assume only point-to-point communication, and consider both synchronous and asynchronous communication. Generally speaking, in this model one can achieve stronger results in two aspects: first, efficient deterministic solutions are possible, and second, the problem is solvable even in fully asynchronous environments.

\paragraph{Crash faults in synchronous networks.}  When relaxing the failures to crashes, one can overcome the lower bound of $\Omega(\byzfrac n)$ that applies to deterministic \download\ under Byzantine faults while managing to tolerate an \emph{arbitrary} fraction $\byzfrac < 1$ of crashes. We first show a simple query optimal protocol that archives $\Query=O(n/\goodfrac k)$, $\Time=O\left((n+\byzfrac k)\byzfrac k\right)$ and $\Message=O(n\byzfrac k^2)$. Then, we show a more complicated and carefully constructed protocol that also achieves query optimality and improves the time and message complexity to $\Time=O(n+\byzfrac k)$, $\Message=O\left((n+\byzfrac k) \cdot k\right)$.

\paragraph{Crash faults in asynchronous networks.}
Another significant distinction of the Crash fault model is that it allows query optimal \emph{asynchronous} protocols. We show a query optimal protocol that achieves $\Query = O(n/k)$, $\Time = O(n)$, and $\Message=O(nk^2)$.

\begin{table}[!t]
\centering \footnotesize
\begin{tabular}
{|P{3.68cm}|P{2.15cm}|P{2.90cm}|P{2.59cm}|P{2.9cm}|}

\hline
\multicolumn{5}{|c|} {\bf Comparison of the Existing and Developed Results for Byzantine Fault}\\
\hline
  {\bf Adversary}&{\bf Theorem} &{\bf Query} & {\bf Time} & {\bf Message Size}\\
\hline

                 & \cite{ABMPRT24} & $\tilde{O}(n/k + \sqrt{n})$  & $O(n)$        & $O(\log n)$\\
Fixed Byzantine  & \cite{ABMPRT24}\# & $\tilde{O}(n/k)$             & $O(n)$        & $O(\log n)$\\
                 & Thm~\ref{thm:opt-alg-p2p} & $\tilde{O}(n/k)$     & $O(n\log k)$  & $O(1)$\\
\hline

Dynamic Byzantine   &Thm~\ref{thm:2-round}         &$\tilde{O}( n/ k + \sqrt{n})$     &$2$ & $\tilde{O}(n/ k + \sqrt{n})$\\
           &Thm~\ref{thm:fast_expected_optimal}   &$\tilde{O}( n/ k)$*      &$O(\log n)$ & $\tilde{O}(n)$\\
\hline

Dynamic Byzantine &Cor~\ref{cor:broadcast-time-query} &$O(n/k)$* &$O(\log n)$& $O(n/k)$\\
and Broadcast     &Thm~\ref{thm:broadcast_whp}    &$O(n/k)$   & $O(\log^2 n)$ &$\tilde{O}(n/k)$\\

\hline
\end{tabular}
\caption{Our main results (with $\beta$ treated as {\it any} positive constant in $[0,1)$). Here, $^\ast$ denotes results that only hold in expectation, and \# denotes that $\beta<1/3$. }\label{tbl: comparative_analysis}
\end{table}
\begin{table}[!t]
\centering \footnotesize
\begin{tabular}
{|P{3.68cm}|P{2.15cm}|P{2.90cm}|P{2.59cm}|P{2.9cm}|}

\hline
\multicolumn{5}{|c|} {\bf Deterministic Protocol with Crash Fault}\\
\hline
  {\bf Model}&{\bf Theorem} &{\bf Query} & {\bf Time} & {\bf Message Size}\\
\hline

Synchronous &Thm~\ref{thm:sync_rapid_download}\# &$O(n/k)$ &$O(n+k)$& $O(\log n)$\\
\hline

Asynchronous &Thm~\ref{thm:async_single_crash} &$O(n/k)$ & $\tilde{O}\left(\frac{n}{k}\right)$& $O(\log n)$\\
             &Thm~\ref{thm:async_multi_crash} & $O(n/k)$ &$O(n)$& $O(\log n)$\\

\hline
\end{tabular}
\caption{Our main results for $\beta<1$ in the crash fault setting for synchronous and asynchronous model. \# denotes that $\beta \leq 1/ k$.}\label{tbl: sync_results}
\end{table}
\section{Query-optimal \download} \label{sec: optimal_time_P2P}

In this section, we show how to optimize the query complexity of \download\ up to a factor of $\log{n}$, achieving $\Query  =  \optqexact  =  \optqtilde$ for any $\byzfrac < 1$. This provides an improvement over the best previously known results for \download~\cite{ABMPRT24}, which either guaranteed a query complexity of $\Query = \tilde{O}\left(\frac{n}{\goodfrac k} + \sqrt{n}\right)$ or imposed the additional restriction of $\byzfrac < \frac{1}{3}$.

The algorithm described next works under the assumption that $2\goodfrac k > 2^\delta \cdot \lg^2{n}$, for some constant $\delta$.\footnote{For any positive real number $x$, we use the notation $\lg{x}$ to denote $\log_2{x}$.}  Note that, for smaller values of $k$, the desired bound of $\optqtilde$ on the query complexity holds trivially: when $k  \leq  \left(\frac{2^{\delta}}{2\goodfrac}\right) \cdot \lg^2{n}$, this bound is $\optqtilde = \tilde{O}(n)$, which is attainable by the trivial algorithm where every peer queries the entire input vector.

Our algorithm runs in $n$ \emph{epochs}. For $1 \leq i \leq n$, a non-faulty peer $M$ dedicates the \ith epoch to learning the \ith input bit. Furthermore, each epoch $i$ is divided into \emph{rounds}. During each round $j$ of the epoch $i$, $M$ attempts to learn the \ith input bit $x_i$ in one of two ways:
\begin{enumerate}
    \item[(i)]  based on messages received from the other peers in earlier rounds, or

    \item[(ii)] by querying the source directly (which happens with gradually increasing probability).
\end{enumerate}
$M$ relies on the messages of its peers for learning $x_i = b$ in round $j$ only if the number of peers who sent it the value $b$ exceeds a specified threshold (depending on $j$) \emph{and} the number of peers who sent it the value $(1 - b)$ is below that threshold. If $M$ fails to learn the bit in round $j$, then it proceeds to round $(j + 1)$, in which it doubles its probability of querying. This is repeated until $M$ successfully learns $x_i$. (Note that a head is always thrown when $2^j$ exceeds $\frac{\goodfrac k}{\lg{n}}$.

Once $M$ learns $x_i$, it broadcasts it and blacklists any peers who send contradictory values. In subsequent rounds, $M$ never relies on bits sent by blacklisted peers. See Algorithm~\ref{alg:improved} for the pseudocode. Note that in each epoch $i$, $M$ learns the value of the \ith bit either by \emph{gossip learning} (i.e., receiving messages with a decisive majority, in line \ref{l:gossip learning} of the algorithm) or by \emph{query learning} (in line \ref{l:query learning}). We name the epoch accordingly as a \emph{gossip} epoch or as a \emph{query} epoch. Furthermore, in every epoch $i$, $M$ performs a ``learning step'' in \emph{exactly one} round, hereafter referred to as the \emph{learning round} of epoch $i$ and denoted $\ell(i)$.

\begin{algorithm}
\caption{Resilient \download\ for any fixed known $\byzfrac<1$; code for \peer $M$} \label{alg:improved}
    
    \begin{algorithmic}[1]
        \State $B \gets \emptyset$ \Comment{Set of blacklisted \peers}
        \For{$i$ from 1 to $n$} \label{for:outer} \Comment{Epoch $i$ deals with input bit $x_i$} 

            \State $I_{voted} \gets 0$ \Comment{Indicator of learning the $i$-th bit and voting on it}
            \State $COUNT^0_i \gets 0$; $COUNT^1_i \gets 0$; \Comment{Counters for the number of \peers that voted 0/1 for bit $i$}
            \State $f \gets\lceil{\delta + \lg \lg n}\rceil$; \Comment{See Lemmas~\ref{lem:correctness} and \ref{lem:complexity} for the choice of the constant $\delta$. } \label{lno:constant}
            \For{$j$ from $f$ to $\lceil{\lg (\gamma k)-\lg \lg n} \rceil $} 
            \label{for:inner}
                \If{$I_{voted}=0$ ($M$ has not voted for bit $i$ yet)}
                        \If{$COUNT^{1-b}_i< \nu \cdot 2^{j}$ for $\mbox{$\nu=1/4$ and}$ for some $b \in \{0,1\}$ AND $(j \neq f)$ }
                   \label{l:insufficient_opposition}   \State $b_i^M \gets b$ 
                    \label{l:gossip learning}               \Comment{Perform a “gossip learning” step}
                      \State Broadcast a vote $\langle M,i,b_i^M=b \rangle$\; \Comment{$M$ and $i$ are implicit}
                       \State Set $I_{voted}\gets 1$
                        \Else
                \If{$j = \lceil {\lg (\gamma k-\lg \lg n}\rceil$} 
                \State Set coinflip to heads
                \Else\label{line: else_part}
                    \State Toss $2^j$ independent Coins with bias $1/\gamma k$ towards heads.
                    \State Let $S_{i,j}$ be the number of Coins that turned heads.
                    Set coinflip to heads if $S_{i,j} \ge 1$.
                    \label{lno:toss}
                \EndIf
                    \If{heads} \Comment{Perform a ``query learning" step}
                            \State  $b_i^M \leftarrow b \leftarrow query(i)$; \label{l:query learning} \Comment{$M$ learns $b^M_i=b$ by query}
                        \State Broadcast a vote $\langle M,i,b_i^M=b \rangle$\; \Comment{$M$ and $i$ are implicit}
                        \State Set $I_{voted}\gets 1$
                        \EndIf
                        \EndIf
                         \State Receive messages from other \peers.
                         \For{$b \in \{0,1\}$}
                        \State $COUNT^b_i \gets |\{M' \not\in B \mid M ~\hbox{received message~}\langle M',i,b_i^{M'}=b \rangle ~\hbox{during epoch $i$}\}|$
                        \\
                        \Comment{count voters for $b \in \{0,1\}$; ignore peers
                        blacklisted in earlier epochs}
                	\EndFor
                \EndIf
                \If{$I_{voted}=1$} \Comment{Including the case when $I_{voted}$ changed during the current phase}
                	\State $B\leftarrow B ~ \cup \{M' \mid \hbox{$M'$ sent message with~} b^{M'}_i \ne b^M_i\}$ \Comment{peers with  contradictory vote}
                \EndIf
            \EndFor
        \EndFor
    \end{algorithmic}
\end{algorithm}

\paragraph{Correctness.}  Let $P_j  =  1 - \left(1 - \frac{1}{\goodfrac k}\right)^{2^j}$ be the probability that when flipping $2^j$ independent random coins, each with bias $\frac{1}{\goodfrac k}$ toward head, at least one turns heads. The following technical lemma provides upper and lower bounds on $P_j$.
\begin{lemma} \label{lem: bound Pj - 2}
    For $j \in [f,\lg (\goodfrac k) - \lg \lg n ]$, $P_j \in \left(\frac{2^j}{\goodfrac k} \left(1 -\frac{1}{2\lg n}\right), \frac{2^j}{\goodfrac k}\right)$. 
\end{lemma}
\begin{proof}
The binomial expansion of $(1 - \varepsilon)^m$ yields the bounds
$1 - m\varepsilon  <  (1-\varepsilon)^m  <  1 - m\varepsilon + \frac{m(m-1)}{2}\varepsilon^2$. Setting $\varepsilon=1/\gamma k$ and $m=2^j$, we get 
$1-2^j/\gamma k < (1-1/\gamma k)^{2^j} < 1-2^j/\gamma k + 2^j(2^j-1)/(2(\gamma k)^2)$, or
$$\frac{2^j}{\gamma k} ~>~ P_j ~>~ \frac{2^j}{\gamma k} - \frac{2^j(2^j-1)}{2(\gamma k)^2}
~>~ \frac{2^j}{\gamma k}\left(1 - \frac{1}{2} \frac{2^j}{\gamma k}\right)
~>~ \frac{2^j}{\gamma k}\left(1 - \frac{1}{2\lg n}\right)
~,$$
where the last inequality follows from the assumption that $2^j < \gamma k/\lg n$.
\end{proof}

To prove the correctness of the algorithm, we need to show that all \nonfaulty \peers $M$ compute $\invec^M = \invec$ correctly with high probability. This is shown by inductively proving that, for any \nonfaulty \peer $M$ and any bit $i$, $1 \le i \le n$, $b_i^M = b_i$ with high probability. The induction is on $i$ for a fixed \nonfaulty \peer $M$. We prove the following two invariant statements inductively.

\begin{description}
    \item[Blacklisting Statements ($BL_i$).]  All \peers blacklisted by some \nonfaulty \peer $M$ until the end of the $i$-th epoch (i.e., the \ith iteration of the outer {\tt for} loop, line~\ref{for:outer}) of Algorithm~\ref{alg:improved}) are indeed Byzantine, with high probability. For convenience, $BL_0$ refers to the empty blacklist before the execution enters the loop and is clearly true.
    
    To capture the inner {\tt for} loop as well (line~\ref{for:inner} of Algorithm~\ref{alg:improved}), we use $BL_{i,j}$, $1 \le i < n$ and  $j \in [f, \lceil{\lg (\gamma k)-\lg \lg n} \rceil ]$ to refer to the statement that all \peers blacklisted by some \nonfaulty \peer $M$ until the end of the \jth inner {\tt for} loop of the $(i+1)$-st epoch are indeed Byzantine, with high probability. For convenience, we sometimes use $BL_{i, f-1}$ to refer to $BL_i$. Also, note that $BL_{i,\lceil{\lg (\gamma k)-\lg \lg n} \rceil}$ implies $BL_{i+1}$.

    \item[Correctness Statement ($C_i$).] All \nonfaulty \peers have executed the first $i$ epochs correctly, with high probability, i.e., $b_{i'}^M = b_{i'}$ for all $1 \le i' \le i$ and all \nonfaulty \peers $M$. Again, $C_0$ refers to the execution being vacuously correct before any bit is processed by the \peers. Furthermore, $C_{i,j}$, $1\le i<n$ and $j \in [f, \lceil{\lg (\gamma k)-\lg \lg n} \rceil ]$,  refers to the following statement at the end of iteration $j$ in epoch $(i+1)$: for all non-faulty peers $M$, $b_{i'}^M = b_{i'}$ for all $1 \le i' \le i$ and $b_{i+1}^M = b_{i+1}$ for all \nonfaulty \peers that have set their respective $I_{voted}$ bit to 1, with high probability. For convenience, we sometimes use $C_{i, f-1}$ to refer to $C_i$. Also, note that $C_{i,\lceil{\lg (\gamma k)-\lg \lg n} \rceil}$ implies $C_{i+1}$.
\end{description}

Observe that \peer $M$ performs a query if (a) it is the first round of that particular epoch and it has obtained ``heads'', or (b) it is round $j > f$, it is the first time $M$ has drawn ``heads'', and $COUNT_i^b \ge \nu 2^j$ for \emph{both} $b=0$ and $b=1$.

The basis for the induction is provided by the statements $BL_0$ and $C_0$, which are both true. The inductive step is given by the following lemma.
\begin{lemma} \label{lem:correctness}
    Consider some epoch $i \in [1,n]$ and assume that conditions $BL_{i-1}$ and $C_{i-1}$ hold at the start of epoch $i$. For any fixed $c \ge 1$, there is a suitably small fixed choice for the constant $\delta$ such that the statements $BL_i$ and $C_i$ hold with probability at least $1 - \frac{1}{n^{c + 1}}$ at the end of the epoch $i$. 
\end{lemma}
\begin{proof}
Fix $c \ge 1$. To prove the claim, we show that $BL_{i,j}$ and $C_{i,j}$ hold for all rounds $j \in [f, \lceil{\lg (\gamma k)-\lg\lg n} \rceil]$. By the inductive hypothesis, $BL_{i, f-1} = BL_{i-1}$ and $C_{i, f-1} = C_{i-1}$ hold. For the sake of contradiction, suppose statement $C_i$ does not hold, and let $j^*$ be the first round $j$ in epoch $i$ when statement $C_{i-1,j}$ is false, namely, some \peer\ $M$ tosses heads and goes on to set its $I_{voted}$ bit to 1, but incorrectly assigns $b^M_i = 1-b_i$. Since $j^*$ is the first such occurrence of incorrect behavior, $BL_{i-1,j^*-1}$ and $C_{i-1,j^*-1}$ are true. If $j^* = f$, then $M$ will explicitly query the bit, so $b_{i}^M = b_{i}$, a contradiction. So, we focus on the case where $j^* > f$.

Without loss of generality, let $b_i = 0$. We claim that the number of votes received by $M$ in favor of $b_i = 0$ will be at least  $\nu \cdot 2^{j^*} =2^{j^*-2}$ with  probability at least $1 - \frac{1}{n^c}$. Since $BL_{i-1,j^*-1}$ and $C_{i-1,j^*-1}$ are true, all \nonfaulty \peers that tossed heads in earlier rounds $j < j^*$ would have correctly voted for $b_i = 0$. To establish the required contradiction, we show that the number of such votes is at least $\nu \cdot 2^{j^*} = 2^{j^*-2}$. Let $G_{old}$ be the set of votes received by $M$ from \nonfaulty \peers in epoch $i$ during rounds $j$ for $j \le j^*-2$ (if such rounds exist). Let $G_{recent}$ be the set of votes received by $M$ from \nonfaulty \peers in round $j' = j^*-1$; we know this round exists as $j^*>f$. Our goal is to show that $X = G_{old} \cup G_{recent}$ has cardinality $|X| \ge 2^{j^*-2}$ with high probability. Let $G$ denote the set of all \nonfaulty \peers; $|G| \ge \goodfrac k$. Note that for every \peer in $G \setminus G_{old}$, the probability of joining $G_{recent}$ is $P_{j'} = P_{j^*-1}$ by Lemma {\ref{lem: bound Pj - 2}}, so the expected size of $X$ is
$$\Exp[|X|] ~\ge~ |G| \cdot P_{j^*-1} ~\ge~ \gamma k \cdot \frac{2^{j^*-1}}{k\gamma} \left(1 -\frac{1}{2\lg n}\right) ~=~ 2^{j^*-1} \left(1 -\frac{1}{2\lg n}\right) ~\ge~ 0.9\cdot 2^{j^*-1},$$
for sufficiently large $n$ (say, $n\geq 32$).
Thus, applying Chernoff bound, we get
\begin{eqnarray*}
\Pr[|X| &<& 2^{j^*-2}] ~=~ \Pr\left[|X| < 0.5\cdot 2^{j^*-1}\right] 
~\le~ \Pr\big[|X| < (1-4/9) \Exp[|X|]\big]
\\
&<& e^{-\frac{(4/9)^2}{2} \Exp[|X|]} 
~\le~ e^{-\frac{(4/9)^2}{2} \cdot 0.9\cdot 2^{j^*-1}}
~\le~ e^{-\frac{4}{45}\cdot 2^f}
~=~ e^{-\frac{4}{45}\cdot 2^{\delta+\lg\lg n}}
~=~ e^{-\frac{4}{45}\cdot 2^\delta \cdot \lg n}
\\
&< & n^{-\frac{4}{45}\cdot 2^\delta}
~< ~ \frac{1}{n^{c+2}},
\end{eqnarray*}
where the last inequality holds when $\delta \ge \lg (45(c+2)) - 2$).

Thus, the number of votes received by $M$ for the correct bit value (0, as per our assumption wlog) must be at least $\nu \cdot 2^{j^*}$ relying on the fact that by the inductive assumption $BL_{i-1,j^*-1}$, no \nonfaulty \peer was added to $B$. From the else part (see line \ref{line: else_part} of Algorithm~\ref{alg:improved}), if both bit values received $\nu \cdot 2^j$ votes (or more), then $M$ will query bit $i$ and this will ensure $C_{i-1, j^*}$. On the other hand, if the number of votes for bit value 1 is less than $\nu 2^{j^*}$, then, $M$ will vote for 0, which will be the correct bit. Thus, the contradiction is established, thereby implying $C_i$ with probability at least $1 - \frac{1}{n^{c + 1}}$ (applying the union bound over all peers $M$).

If $C_i$ holds, then all other \nonfaulty \peers $M' \ne M$ also vote correctly for $0$, so $M$ will not blacklist any of them. This implies that the invariant $BL_i$ also holds.
\end{proof}

Applying Lemma~\ref{lem:correctness} and taking the union bound over all $n$ epochs, we get the following corollary.
\begin{corollary}
    For $\delta$ fixed as in Lemma~\ref{lem:correctness}, Algorithm~\ref{alg:improved} ensures that all \peers $M$ correctly compute $\invec^M = \invec$ with probability at least $1 - \frac{1}{n^c}$.
\end{corollary}

\paragraph{Query complexity analysis}  In the absence of Byzantine peers, the total number of (necessary) queries is $O(n)$, so the average cost per \peer is clearly $\Query = O\left(\frac{n}{\goodfrac k}\right)$. The reason for the additional wasteful queries is two-fold. First, the fact that the algorithm is randomized and must succeed in learning all bits with high probability requires some redundancy in querying. Second, Byzantine \peers spread fake information, forcing \nonfaulty \peers to perform queries to blacklist the culprits and clarify the true values of $\invec$.

Let $\randvec(M)$ denote the sequence of independent random coins flipped by \peer $M$ during execution (each with a bias of $\frac{1}{\goodfrac k}$ towards heads). For notational convenience, we will often write $\randvec$ \st{$(M)$} to mean $\randvec(M)$, when the underlying peer $M$ is clear from the context. Let $\randvec_{i,j}$ denote the subsequence of $\randvec$ that was drawn at the beginning of round $j$ of the epoch $i$, and let $\randvec_i$ denote the subsequence of $\randvec$ that was drawn during the entire epoch $i$. Respectively, let $R = \modof{\randvec}$, $R_i = \modof{\randvec_i}$ and $R_{i,j} = \modof{\randvec_{i,j}}$.

The variable $S_{i, j}$ used in the algorithm denotes the number of coins of $\randvec_{i,j}$ that turned heads. Similarly, let $S_i$ denote the number of coins of $\randvec_i$ that turned heads.
\begin{lemma} \label{lem:complexity}
    For any fixed $c' \ge 1$, there is a suitably small fixed choice for the constant $\delta$ such that Algorithm~\ref{alg:improved} has query complexity $\Query = \optqexact$ with probability $1 - \frac{1}{n^{c'}}$.
\end{lemma}
\begin{proof}
    Throughout the proof, we consider a non-faulty peer $M$.

    For every $1\le i\le n$, let $B(i)$ denote the number of \peers that were blacklisted by $M$ in epoch $i$ and let $C(i)$ denote the number of queries performed by $M$ in epoch $i$. (Note that $C(i)$ is 0 if $i$ is a \emph{gossip epoch} and 1 if $i$ is a \emph{query epoch}.)

    
    As $C(i)=0$ for a gossip epoch $i$, the total cost of gossip epochs is
    \begin{equation} \label{eq:Cgossip}
        C^{gossip} ~=~ \sum_{\mbox{\footnotesize gossip}~i} C^{gossip}(i) ~=~ 0.
    \end{equation}
    Hence, we only need to analyze query epochs. We bound the number of queries in these epochs by first bounding the number of random coins flipped in these epochs.
    
    Consider such an epoch $i$, with learning round $\ell = \ell(i)$. There are two cases. The first is when $\ell = f$. In this case, perhaps no faulty \peers were blacklisted, so $B(i)\ge 0$ and the number of random Coins flipped in this epoch is $R_i = R_{i,f} = 2^f$.
    
    The second case is when $\ell > f$. The fact that $M$ had to query in round $\ell$ implies that gossip learning was not possible, so both $COUNT^0_i \ge \nu 2^j = 2^{\ell-2}$ and $COUNT^1_i\ge 2^{\ell-2}$. This, in turn, implies that the number of faulty \peers that $M$ gets to blacklist in this epoch is $B(i) \ge 2^{\ell-2}$. On the other hand, the number of Coins used during this epoch is $R_i = \sum_{j=f}^\ell R_{i,j} = 2^{\ell+1}-2^f$.
    
    Combining both cases, we get that for every epoch $i$, $B(i) \ge (R_i-2^f)/8$. Summing over all $i$, we get that
    \begin{equation*}
        \beta k ~\ge~ B ~=~ \sum_i B(i) ~\ge~ \sum_i \frac{R_i-2^f}{8} ~=~ \frac{R-2^fn}{8}~.
    \end{equation*}
    Rewriting, and recalling that $f=\lceil{\delta + \lg \lg n}\rceil$, we get
    \begin{equation} \label{eq:bound R}
        n\cdot 2^\delta\lg n ~\le~ R ~\le~ 8\beta k + 2^f n~\le~ 8\beta k + n \cdot 2^\delta \lg{n}\text{,}
    \end{equation}
    where the first inequality follows from the fact that in every epoch $i$, $R_i\ge R_{i,f} \ge 2^\delta\lg n$.
    
    Note that $C(i)\le S_i$ for every $i$, hence
    \begin{equation} \label{eq:Cquery}
        C^{query} ~=~ \sum_{\mbox{\footnotesize query}~i} C(i) \le \sum_{\mbox{\footnotesize query}~i} S_i ~=~ S.
    \end{equation}
    Recalling that $\randvec$ is a sequence of $R$ independent Bernoulli variables with probability $\frac{1}{\goodfrac k}$ for heads, we have that $\Exp[S]  =  \frac{R}{\goodfrac k}$, and applying Chernoff bound we get
    \begin{equation} \label{eq: S vs R}
        \Pr[S > 2R/\gamma k] ~=~ \Pr[S > 2 \Exp[S]] ~\le~ e^{-\Exp[S]/3} ~=~ e^{-R/3\gamma k} ~\le~ e^{-n 2^\delta\lg n/3\gamma k} ~\le~ \frac{1}{n^{c'}}~\text{,}
    \end{equation}
    where the penultimate inequality relies on the left side of Eq. \eqref{eq:bound R} and the last one holds when $\delta \ge \lg (3\gamma c')$, also relying on the assumption that $k \le n$. Combining Equations \eqref{eq:Cquery}, \eqref{eq: S vs R}, and the right side of \eqref{eq:bound R}, we get that with probability at least $1 - \frac{1}{n^{c'}}$,
    \begin{equation} \label{eq:Cquery-sum}
        C^{query} ~\le~ S ~\le~ \frac{2R}{\gamma k} ~\le~ \frac{2(8\beta k + cn\lg n)}{\gamma k}  ~=~  \optqexact.
    \end{equation}
    Finally, the lemma follows by Equations {\eqref{eq:Cgossip}} and {\eqref{eq:Cquery-sum}}.
\end{proof}

Each epoch lasts $O(\lg{k})$ phases, so the time complexity is $O(n\lg{k})$. Each \peer sends at most $n$ bits to each of the other machines, so the total number of messages sent is $O(nk^2)$. Thus, we have the following theorem.
\begin{theorem} \label{thm:opt-alg-p2p}
    In the synchronous point-to-point model with Byzantine failures, there is a randomized algorithm for \download\ such that with high probability, $\Query = \optqexact$, $\Time = O(n\log{k})$, and $\Message = O(nk^2)$. Moreover, it only requires messages of size $O(1)$.
\end{theorem}
\section{A Lower Bound} \label{section-lower-bounds}

Having established the existence of a query-optimal algorithm for the \download\ problem with an arbitrary bound $\byzfrac<1$ on the fraction of faulty peers, we turn to the issue of time complexity and derive faster \download\ algorithms. In this section, we ask whether a single round algorithm exists (assuming arbitrarily large messages can be sent in a single round), other than the trivial algorithm where every \peer queries every bit. We show that this is impossible: if we insist on a single round algorithm, then every \peer must query all bits.

\paragraph{Notations and Terminology.}  Let $V$ be a set of $n$ \peers in the network. Recall that $\beta$ is the maximum allowed fraction of Byzantine \peers in any execution and $\gamma$ is the minimum fraction of \nonfaulty \peers ($\beta+\gamma=1$). In this section we assume, for simplicity, that $n$ is odd, and fix $\beta=(n-1)/2n$, i.e., we assume that at most $(n-1)/2$ \peers might fail in any execution.
Let the input vector be denoted by $\invec  =  \left\{b_1, b_2, \ldots, b_n\right\}$.

For positive integers $q$ and $t$, let $\mathcal{A}(q, t)$ be the class of all randomized Monte-Carlo algorithms with query complexity at most $q$ and time complexity at most $t$. (Note that this class encompasses the class of deterministic algorithms as well.) We fix an algorithm $A \in \mathcal{A}(n - 1, 1)$, and from this point on we make our arguments with respect to this algorithm $A$.

In every execution of algorithm $A$, the first (and only) round starts (at the beginning of the query sub-round) with each \peer $v\in V$ making some random coin flips and generating a random bit string $R_v$. The resulting \emph{random profile} of the execution is $R = \langle R_1, \ldots, R_k \rangle$, the collection of random strings selected by the $k$ peers. Based on its random bit string $R_v$, each \peer $v$ computes a list of indices that it queries (unless it fails).

\begin{defn} \label{def: weak adversary}
    A \emph{weak} adversary is one that fixes its adversarial strategy \emph{before} the start of the execution. This strategy is composed of selecting a set $\byz\subseteq V$ of \peers to fail in the execution, subject to the constraint that $|\byz| \leq \beta n$, and a set of rules dictating the actions of the \peers in $\byz$, including the messages that are to be sent by each Byzantine \peer. 
\end{defn}

\begin{defn} \label{def: execution}
    We define an \emph{execution} of the algorithm $A$ as the 3-tuple $\EXEC = (\invec, \byz, R)$, where $\invec$ is the input vector (which can be chosen by the adversary), $\byz$ is the set of Byzantine \peers selected by the adversary, and  $R$ is the random profile selected by the \peers. We denote $\good=V\setminus \byz$. We say that the execution $\EXEC(\invec,\byz,R)$ \emph{succeeds} if all \peers acquire all $n$ bits. 
\end{defn}

A key observation is that for every non-faulty peer $v$, the list of queried indices is determined by its random bit string $R_v$, and is independent of $\invec$ and $\byz$. Hence, the random profile $R$ fully determines the list of queried indices of every \peer.

We refer to a random profile $R$ for which some \peer fails to query the $i^{\text{th}}$ input bit as an \emph{$i$-defective} random profile. Let
\begin{eqnarray*}
\RR &=& \{ R \mid R ~\mbox{is a random profile of some execution of}~ A \},
\\
\RR_i &=& \{ R \in \RR \mid R ~\mbox{is}~ i\mbox{-defective} \}.
\end{eqnarray*}

Note that each random profile $R$ may occur with a different probability $p_R$, such that $\sum_{R\in\RR} p_R = 1$.
For any subset $\RR' \subseteq \RR$, let
$p(\RR)= \sum_{R\in\RR'} p_R$.

\begin{lemma} \label{lemma-BE-ell}
    For the algorithm $A$, there is an index $\ell=\ell(A)$, $1\le\ell\le n$, such that $p(\RR_\ell) ~\ge~ 1-e^{-1}$.
\end{lemma}

\begin{proof}
Let $\EXSET$ denote the set of executions of algorithm $A$. For each \peer $v$, let $\EXSET_v$ be the set of executions in which $v$ does not fail, and let $r=|\EXSET_v|$. 
For $i \in [1, n]$, let $r_v(i)$ denote the number of executions in $\EXSET_v$ in which $v$ does not query $i$, and let $p_v(i) = r_v(i)/r$ be the probability that $v$ does not query $i$.

Note that by definition of the class $\mathcal{A}(n-1,1)$, $\sum_i r(v,i) \ge r$, because for every execution $e\in \EXSET_v$ there is at least one bit $i$ that $v$ does not query in $e$, or in other words, every such execution is counted in one of the counters $r(v,i)$. It follows that
\begin{equation} \label{eq:P ge 1}
    \sum_{i = 1}^n p_v(i) ~=~ \frac{\sum_i r(v,i)}{r} ~\geq~  1 ~~~~~~ \mbox{for every \peer}~ v.
\end{equation}
The adversary can calculate the probabilities $p_v(i)$
for every \peer $v$. Subsequently, it can compute for every index $i$ the probability that \emph{no} \peer skips $i$, namely, 
\begin{equation} \label{eq:qi formula}
    q(i) ~=~ \prod_{v} (1 - p_v(i)),
\end{equation}
and pick $\ell$ to be the index minimizing $q(i)$. We claim that for such an index $\ell$, $p(\RR_\ell) ~\ge~ 1-e^{-1}$. To see this, note that it can be shown that under constraints {\eqref{eq:P ge 1}} and {\eqref{eq:qi formula}}, $q(i)$ is maximized when $p_v(i) = \frac{1}{n}$ for every $v$ and every $i$. Then
\begin{equation*}
    q(i) ~=~ \prod_{v} (1 - p_v(i))   ~\leq~   \left(1 - \frac{1}{n}\right)^{n}
    ~<~   e^{-1}~.
\end{equation*}
As $p[\RR_\ell] = 1-q(\ell)$, the lemma follows.
\end{proof}

For an $\ell$-defective random profile $R\in \RR_\ell$, denote by $v(R)$ the \emph{defective} \peer $v$ that does not query the $\ell^{\text{th}}$ input bit. (In case more than one such \peer exists, we pick one arbitrarily.)
\begin{lemma}
    There is a (non-adaptive) adversarial strategy that foils Algorithm $A$ with probability at least $(1-e^{-1})/4$.
\end{lemma}
\begin{proof}
The proof is by an \emph{indistinguishability} argument. 
Let $\ell=\ell(A)$ be the index whose existence is asserted by Lemma \ref{lemma-BE-ell}.
For any input vector $\invec$, let $\invX{\invec}$ be another $n$-bit vector that's identical to
$\invec$ except for the $\ell^{\text{th}}$ bit, which is inverted in $\invX{\invec}$. The adversary fixes an arbitrary input vector $\invec_0$ in which $b_\ell=0$,
and sets the vector $\invec_1$ to be  $\invX{\invec_0}$.
The adversarial strategy is to select the set $\byz$ of $(n-1)/2$ Byzantine \peers uniformly at random from the collection $\BB$ of all $n\choose (n-1)/2$ possible choices. For any input vector $\invec$, the Byzantine peers are instructed to execute the same protocol as the \nonfaulty \peers, but on the vector $\invX{\invec}$. In other words, in executions on $\invec_0$, the Byzantine \peers will run algorithm $A$ on $\invec_1$, and vice versa.

We are interested in subsets of the set $\EXSET$ of executions of algorithm $A$. In particular, for any specific input $\invec'$ and subsets $\RR'\subseteq\RR$ and $\BB' \subseteq \BB$, let $\EXSET(\invec',\RR',\BB')$ denote the set of all executions $\EXEC(\invec,\byz,R)$ where $\invec=\invec'$, $\byz\in\BB'$ and $R\in \RR'$.

Hereafter, we only consider executions where the random profile $R$ is $\ell$-defective and $v(R)\not\in\byz$ for the set $\byz$ selected by the adversary, i.e.,
where $R$ belongs to the set
$$\hRR(\byz) ~=~ \{ R \mid R\in\RR_\ell,~~~v(R)\not\in\byz\}.$$
Note that since the set $\byz$ is selected by the adversary uniformly at random and independently of $R$, $\Pr[v(R)\not\in\byz] > 1/2$, hence, applying Lemma \ref{lemma-BE-ell}, we get that for any $\byz\in\BB$,
\begin{equation}
\label{eq:p-hRR-byz}
p(\hRR(\byz)) ~>~ (1-e^{-1})/2.
\end{equation}
For the rest of the proof, we concentrate on executions from
$\EXSET(\invec_0,\BB,\hRR(\byz))$ and $\EXSET(\invec_1,\BB,\hRR(\byz))$, and aim to show that in at least one of the two sets, the executions fail with constant probability. 

For $\byz\in \BB$ and random profile $R$, let
$\invB_R = V\setminus(\byz\cup\{v(R)\})$.  

\begin{claim}
\label{cl:one of two fails}
Consider some $\byz\in\BB$ and $R\in\hRR(\byz)$. Then the defective \peer $v=v(R)$ identifies the bit $b_\ell$ in the same way in the executions $\EXEC(\invec_0,\byz,R)$ and $\EXEC(\invec_1,\invB,R)$, namely, at least one of the two executions fails. Formally, noting that both
$\Pr[\EXEC(\invec_0,\byz,R) ~\text{fails}]$ and $\Pr[\EXEC(\invec_1,\invB,R) ~\text{fails}]$ are  0 or 1, the claim implies that 
$$\Pr[\EXEC(\invec_0,\byz,R) ~\text{fails}] + 
\Pr[\EXEC(\invec_1,\invB,R) ~\text{fails}] \ge 1.$$
\end{claim}
\begin{proof}
Let $\goodv$ (respectively, $\byzv$) be the set of messages sent from peers in the set $\good - \lbrace v \rbrace$ (resp., $\byz$) to $v$ in execution $\EXEC(\invec_0,\byz,R)$, and let
    $\allv  ~=~  \goodv \cup \byzv$.
Define $\goodpv$, $\byzpv$ and $\allpv$ analogously
w.r.t. the execution $\EXEC(\invec_1,\invB,R)$.
The key observation here is that
$$\goodv = \byzpv ~~~~~ \mbox{ and }~~~~~ \byzv = \goodpv.$$ 
This implies that
\begin{equation}
\label{eq:same messages}
    \allpv  ~=~  \goodpv  \cup  \byzpv  ~=~  \byzv \cup \goodv  ~=~  \allv\text{.}
\end{equation}
That is, the set of messages received by $v$ is identical in both executions.

Since $v$ sees the same inputs and the same incoming messages in both of these executions, it follows that it will identify the bit $b_\ell$ in the same way in both executions. But whichever identification it makes, in one of the two executions this identification is false, hence that execution fails.
\end{proof}

For a given $\byz$, let $p_{fail}^0(\byz,\RR')$ denote the probability that the execution $\EXEC(\invec_0,\byz,R)$ fails, conditioned on $R\in \RR'$, and let $p_{fail}^1(\invB,\RR')$ denote the probability that the execution $\EXEC(\invec_1,\invB,R)$ fails, conditioned on $R\in \RR'$. Then
\begin{eqnarray}
p_{fail}^0(\byz,\RR') 
&=& \Pr[\EXEC(\invec_0,\byz,R)~\mbox{fails} \mid R\in\RR']
~=~ \frac{\Pr[\EXEC(\invec_0,\byz,R)~\mbox{fails} ~\wedge~ R\in\RR']}{\Pr[R\in\RR']}
\nonumber
\\
&=& \frac{1}{p(\RR')} \cdot\sum_{R\in\RR'} p_R\cdot\Pr[\EXEC(\invec_0,\byz,R)~\mbox{fails}].
\label{eq:def fail}
\end{eqnarray}

Hence, for any $\byz\in\BB$, applying \eqref{eq:def fail} to $\RR$ and noting that $p(\RR)=1$,
\begin{eqnarray*}
p_{fail}^0(\byz,\RR) &=&
\sum_{R\in\RR} p_R\cdot\Pr[\EXEC(\invec_0,\byz,R)~\mbox{fails}]
~\ge~ \sum_{R\in\hRR(\byz)} p_R\cdot\Pr[\EXEC(\invec_0,\byz,R)~\mbox{fails}].
\end{eqnarray*}
Similarly, we get that
$$p_{fail}^1(\invB,\RR) ~\ge~
\sum_{R\in\hRR(\byz)} p_R\cdot\Pr[\EXEC(\invec_1,\invB,R)~\mbox{fails}].$$
Combining the last two inequalities with Claim \ref{cl:one of two fails} and Eq. \eqref{eq:p-hRR-byz}, we get that for any $\byz\in\BB$,
\begin{eqnarray}
p_{fail}^0(\byz,\RR) + p_{fail}^1(\invB,\RR) &\ge&
\sum_{R\in\hRR(\byz)} p_R\cdot(\Pr[\EXEC(\invec_0,\byz,R) + \Pr[\EXEC(\invec_1,\invB,R))
\nonumber
\\
&\ge& \sum_{R\in\hRR(\byz)} p_R ~=~ p(\hRR(\byz)) ~>~ (1-e^{-1})/2.
\label{eq:pfail 0+1}
\end{eqnarray}
For $i=0,1$, let $p_{fail}^i = \Pr[\EXEC(\invec_i,\byz,R)~\mbox{fails} \mid \byz\in\BB,~ R\in \RR]$. Then
\begin{eqnarray*}
p_{fail}^0 &=& \frac{1}{|\BB|} \sum_{\byz\in\BB} p_{fail}^0(\byz,\RR),
\\
p_{fail}^1 &=& \frac{1}{|\BB|} \sum_{\byz\in\BB} p_{fail}^1(\byz,\RR)
~=~ \frac{1}{|\BB|} \sum_{\byz\in\BB} p_{fail}^1(\invB,\RR),
\end{eqnarray*}
where the last equality follows from observing that, since $((\byz)^{\textsf{inv}})^{\textsf{inv}} = \byz$, and since the probability is taken over all $\byz\in\BB$, $p_{fail}^1$ also equals the probability that the execution $\EXEC(\invec_1,\byz,R)$ fails, conditioned on $\byz\in\BB$ and $R\in \RR$.
It follows from Eq. \eqref{eq:pfail 0+1} that
$$p_{fail}^0 + p_{fail}^1 
~=~ \frac{1}{|\BB|} \sum_{\byz\in\BB} (p_{fail}^0(\byz,\RR) + p_{fail}^1(\invB,\RR))
~>~ \frac{1}{|\BB|} \sum_{\byz\in\BB} (1-e^{-1})/2
~=~ (1-e^{-1})/2.$$
Without loss of generality, $p_{fail}^0 \ge (1-e^{-1})/4$. This implies that the algorithm fails on the input vector $\invec_0$ with this probability. The lemma follows.
\end{proof}

\begin{theorem}
    Consider an $n$-\peer system with external data, for odd $n$, and let the fraction of Byzantine \peers be bounded above by $\beta = (n-1)/2n$. For every algorithm $A \in \mathcal{A}(n - 1, 1)$, there is an adversarial strategy such that $A$ fails to learn all the input bits with probability at least $(1-e^{-1})/4$. The lower bound holds even in the broadcast communication model.
\end{theorem}
\section{Faster \download\ with Near Optimal Query Cost} \label{sec:alg_fast}

In this section, we show how to reduce the time complexity to $O(\log n)$. Unlike the previous protocol, the protocols here do not use blacklisting and do not require fixed IDs from round to round; only that there is a known lower bound on the number of \nonfaulty \peers present in each round.
 
We first describe a key definition and a subroutine in these protocols. 
In a typical step of these protocols, we fix a parameter $\varphi$ depending on the round, and the input vector $\invec$ is partitioned into $\Inum=\lceil n/\varphi \rceil$ contiguous subsets of size $\varphi$, $\invec[\indID,\varphi] = (x_{(\indID-1)\cdot \varphi +1 },x_{(\indID-1)\cdot \varphi +2},\ldots,x_{\indID\cdot \varphi})$, for $\indID \in [1, \Inum]$,
with the last interval, $(x_{(\Inum-1)\cdot \varphi +1 },\ldots,x_n)$), being possibly shorter.
Each nonfaulty peer queries all the bits of some interval $\invec[\indID,\varphi]$ and broadcasts its findings to all other \peers, in the form of a pair
$\langle \indID,s\rangle$, where $s$ is the bit string obtained from the queries.

Hence at the end of a round, every \peer receives a collection of strings for different intervals.
Let $s[j]$ denote the $j$-th bit of the string $s$.
Note that a message $\langle \indID,s\rangle$ can possibly be received in multiple copies, from different peers, and the algorithm keeps all copies.  

Let $S^M$ be a multiset of strings for a particular interval $\indID$ received by

\peer $M$ at the current round.
Define the set of \emph{$t$-frequent} strings at a \peer $M$ as
$$\FS(S,t)=
 \{s\in S ~\mid~ s \hbox{ appears } \geq t \hbox{ times in }S\}$$
Note that the $\FS$ function gets a \emph{multiset} and returns a \emph{set} of $t$-frequent strings. 
A {\it decision-tree} $T$ is a rooted ordered binary tree with the internal nodes labeled by an index of the input array 
and the leaves labeled by string.
Given a set of strings $\FS$ of which one is consistent with the input array (i.e., it appears exactly in the input array at the specified indices), we can build a decision tree to determine which one of its strings is consistent with the input array with a cost of $|\FS|-1$ queries, by querying the indices associated with the internal nodes of the decision tree and traversing the tree accordingly until reaching a leaf, which marks the correct value of the interval (see Algorithm {\ref{alg:decision_tree}}).

\begin{algorithm}
    \caption{\ConstDecTree$(S)$
    \\
    input: Set of strings $S$
    \\
    output: Node labeled tree $T$}
    \label{alg:decision_tree}
\begin{algorithmic}[1]
\State Create a root node $v$
\If{{$|S|>1$}}

    \State Set $i \leftarrow $ smallest index of bit on which at least two strings in $S$ differ
   
    \State $~label(v)\leftarrow i$
       
            \State Set $S^b \gets \{s \in S \mid s[i] = b\}$ for $b\in \{0,1\}$
            \State $T^0\gets$ \ConstDecTree $(S^0)$
            \State $T^1\gets$ \ConstDecTree$(S^1)$
            \State Let $T$ be tree rooted at $v$ with $\leftchild(v)\leftarrow T^0$ and $\rightchild(v)\leftarrow T^1$ 
 
\Else 
    \State $label(v)\gets s$, where $S=\{s\}$ \Comment{$|S|=1$, $v$ is a leaf}
    \State Let $T$ be a tree consisting of the singleton $v$.
\EndIf
\State Return $T$
\\

\Procedure{\Determine}{T}
\State $J =\{j~|~\exists u \in T \;\text{s.t.}\; j=label(u)\}$
\For {all  $j \in J$ } in parallel
\State $query(j)$
\EndFor
\State Let $v$ be the root of $T$
\While{$v$ is not a leaf}
    \State Let $j = label(v)$.
    \If{$b_j =0$}
        \State Set $v\gets \leftchild(v)$
    \Else
        \State Set $v\gets \rightchild(v)$
    \EndIf
    
\EndWhile
\State Return $label(v)$
\EndProcedure
\end{algorithmic}
\end{algorithm}

\subsection{A Simple 2-step Algorithm}
\label{sec: 2-step algorithm}

To illustrate the ideas of the following protocols in this section, we first give a 2-step protocol with $O(n/\sqrt{\gamma k})$ queries which succeeds with high probability. Note that if $\gamma k$ is very small, every nonfaulty peer queries every bit and the algorithm is always correct. Otherwise:
\begin{enumerate}
\item Split the $n$ bit string into $\Inum$ equal length intervals of 
Each \peer picks an interval uniformly at random, queries all its bits, and broadcasts the discovered string $s$ together with the identifier $\indID \in [1, \Inum]$ of its chosen interval. 

\item Let $\FS_\indID$ be the set of strings  each of which was received from at least $t=\gamma k/(2\Inum)$ non-faulty \peers\ for interval $\indID$, disregarding any string sent by a \peer\ which sends more than one string. 
Formally, $\FS_\indID=\FS(S_\indID,t)$ where 
$S_\indID = \{s ~\mid~ \langle \indID,s \rangle \hbox{ received from some peer}\}$.

In parallel, for each interval $\indID \in [1, \Inum]$, build a decision tree $T_\indID$ from $\FS_\indID$ and invoke Procedure \Determine($\cdot$), letting each \peer determine the correct string for interval $\indID$. 

\end{enumerate}

\begin{algorithm}
\caption{2-Round \download\ with $\byzfrac<1$; Code for \peer $M$}
\begin{algorithmic}[1]
\If{$\gamma k< 64 c \ln n $} query every bit and return
\EndIf
\If{$k \geq 12c \ln n \sqrt{2n/\gamma}$}
    $\varphi \gets 16\sqrt{2n/\gamma}$
\Else~ {$\varphi \gets 32 c \ln n (n/(\gamma k))$}
\EndIf
\State $t\leftarrow \gamma k/(2\Inum)$

\State Randomly select  
$\indID^M \in [1, \Inum]$
\State Set string 
$s^M \gets query(\invec[\indID^M, \varphi])$
\State Broadcast $\langle \indID^M,s^M \rangle$
\For{{$\indID=1$} to $\Inum$} in parallel
    \State Construct the multiset $S_\indID \gets \{s ~\mid~ \langle \indID,s \rangle \hbox{ received}\}$
  
\State $T_\indID \gets \ConstDecTree(\FS(S_\indID,t))$

\State $s^\indID \gets \Determine(T_\indID)$
\EndFor
\State 

Output $s^1 s^2 \cdots s^{\Inum}$

\end{algorithmic}
\end{algorithm}


\paragraph*{Correctness.}
Consider an execution of the algorithm, and denote the number of nonfaulty peers that pick the interval $\indID$ by $k_\indID$.
The algorithm succeeds if $k_\indID \ge t$ for every interval $\indID$, 
since every decision tree $T_\indID$ will contain a leaf with the correct string for the interval $\indID$ returned by Procedure \Determine($\cdot$).

\begin{claim}
\label{cl:enough readers}
For constant $c\ge 1$, if $t \geq 8(c+1)\ln n$, then $k_\indID \ge t$ for every interval $\indID \in [1, \Inum]$, with probability at least $1-1/n^c$.

\end{claim}
\begin{proof}
Fix $\indID \in [1,\Inum]$.
The expected number of \nonfaulty \peers\ that pick the interval $\indID$ is 
$\Exp[k_\indID] =\gamma k/\Inum=2t$. 
Therefore, applying Chernoff bounds, 
$\Pr[k_\indID < t] \leq e^{-t/8}$
$\leq 1/n^{c+1}$ for $t \geq 8(c+1)\ln n$. Taking a union bound over all intervals, the probability that 
$k_\indID < t$ for \emph{any} interval $\indID$ is less than $1/n^c$.
\end{proof}


\paragraph*{Query complexity.}
The cost of querying in Step 2 for interval $\indID$ is the number of internal nodes of the decision tree, which is $|\FS_\indID| -1$.

Let $x_\indID$ be the number of strings received for interval $\indID$ in Step 1 (including copies). Then  $|\FS_\indID|\leq x_\indID/t$. Since each peer sends no more than one string overall, 
$\sum x_\indID=k$.
Hence, the cost of determining all intervals is $\sum_\indID |\FS_\indID| \leq  \sum_\indID x_\indID/t \leq k/t$. 

The query cost per \peer, $\Query$, is the cost of determining every interval using decision trees plus the initial query cost; hence $\Query\leq k/t + \varphi$.  Since $t=\gamma k/(2\Inum)$, $ \Query\leq 2 \lceil{n/\varphi}\rceil (1/\gamma) + \varphi $.
To satisfy the premise of Claim {\ref{cl:enough readers}}, it is required that
$t = \gamma k/(2\Inum) \geq  8(c+1) \ln n$.\\

For $k \geq  12(c+1)  \ln n \sqrt{2n/\gamma}$, we set $\varphi=16\sqrt{2n/\gamma}$.
Since $t=\gamma k/(2 \lceil{n/\varphi}\rceil)$, then $t \geq \gamma k/(2n/\varphi +2) \geq  8(c+1) \ln n $ which satisfies the premise of Claim {\ref{cl:enough readers}}. Also $\Query= k/t+ \varphi=O( \sqrt{n/\gamma})$.
Similarly, for $64 c \ln n \leq k < (c+1) \ln n \sqrt{2n/\gamma}$, we set
$\varphi = (32 c \ln n )(n/(\gamma k))$ so that $t \geq 8 (c+1)\ln n$ and $\Query = k/t+\varphi =
O(\sqrt{n/\gamma }) + \varphi = O( \sqrt{n/\gamma} + n \ln n/(\gamma k))$.
Thus,  $\Query=O(\min\{n\ln n/(\gamma k), n\}+\sqrt{n/\gamma})$

\inline Message size:
The message size in this algorithm is the number of bits required to describe the first index of the interval chosen and its bits in Step 1, i.e., $O(\varphi + \log n)=O(\varphi)= O(\sqrt{n/\gamma}+ n\ln n/(\gamma k))$. In other words, the message size is not more than the query complexity in the first round which is not more than $\Query$, i.e., $O(\Query)$.

\begin{theorem}\label{thm:2-round}
There is a 2-round randomized algorithm for \download\ in the point-to-point model with $\Query=O(n \log n/(\gamma k) + \sqrt{n/\gamma})$ and message size $O(n/(\gamma k) \log n+ \sqrt{n/\gamma})$.
\end{theorem}

\begin{remark}\label{remark: optimal_time_query}
     For $k< O(\sqrt{n/\gamma}\log n)$, time complexity is optimal whereas query complexity is optimal up to a $\ln n$ factor. In this, $k$ may vary till $n$ based on the value of $\gamma$.
\end{remark}

\subsection{\download\ with \texorpdfstring{$O(n\log n/(\gamma k))$}{p} Expected Queries and \texorpdfstring{$O(\log n)$}{log n} time}
\label{ss:expected}

The 2-step protocol presented in the previous subsection, while fast and simple, is not query-optimal, with its main source of overhead being that every \peer, after its initial query of the randomly selected interval, builds a decision tree for every other interval and determines the correct leaf (and hence the correct string) by performing additional queries. In this subsection, we extend the previous protocol, achieving optimal query complexity at the cost of going from $O(1)$ to $O(\log n)$ rounds.

The algorithm is based on the following idea. Assume for simplicity that $n$ is a power of 2 and let 
$\varphi$ $= (n/(\gamma k)) (8 (c+1) \ln n )$. 
Recall that $\Inum=\lceil n/\varphi \rceil$.
In Step $i$, for $i \in [0,\dots,\log \Inum]$, the $n$ bits are partitioned into 
$\Inum_i = \lceil n/\varphi_i \rceil$
intervals of size $\varphi_i=2^i \varphi$. Each interval in step $i$ consists of the concatenation of two consecutive intervals from step $i-1$, beginning with the first two intervals. In each step, each \nonfaulty \peer\ determines a random interval of increasing size by splitting it into its two interval parts and using the technique described in the 2-step protocol above on both parts, until the entire input vector is determined. This significantly reduces the number of decision trees constructed by each \peer, from $\Inum$ to $O(\log n)$, which allows the improved query complexity.
Schematically, the protocol works as follows.

See Algorithm {\ref{alg:time-query}} for a formal description.

\commentstart
\begin{algorithm} 
\caption{\download\ Protocol with $\tilde{O}(n/\gamma{k})$  Expected Queries and ${O(\log n)}$ Time $\byzfrac<1$}
\label{alg:time-query}
\begin{algorithmic}[1]
\State Randomly pick an interval $\indID \in [1,\Inum]$ of size $\varphi$
\State Set string 

$s=query(\invec[\indID, \varphi])$
\Statex broadcast$\langle \indID,s,0\rangle$
\For{$i=$ 1 to $\lg\Inum$}
\State Randomly pick an interval 
$\indID \in [1, \Inum_i]$ of size 
$\varphi_i$

\State $\indID_L \gets 2\indID-1$;$\indID_R \gets 2\indID$
\For{$u \in \{\indID_L, \indID_R\}$} in parallel 
\State Construct the multiset $S(u) \leftarrow \{s~\mid~  \text{received a message of the form }\langle u,s,i-1\rangle\}$ 

\State $t_{i-1}=2^{i-2} \varphi (\gamma k/n)$

\State $T_u \gets \ConstDecTree(\FS(S(u),t_{i-1}))$ \label{line:construct}
\State $s_{u} \gets \Determine(T_u)$

\EndFor
\State Set $s_{\indID}$ to $s_{\indID_L}s_{\indID_R}$ and broadcast $\langle \indID,s_{\indID},i\rangle$
 \EndFor  
\State return the determined string for interval $[1,\dots, n]$
\end{algorithmic}
\end{algorithm}
\commentend

\begin{algorithm} 
\caption{\download\ Protocol with $\tilde{O}(n/\gamma{k})$  Expected Queries and ${O(\log n)}$ Time $\byzfrac<1$}
\label{alg:time-query}
\begin{algorithmic}[1]

\For{$i=$ 0 to $\lg\Inum$}
    \State Randomly pick an interval $\indID \in [1, \Inum_i]$ of size $\varphi_i$
    \If{$i=0$}
        \State Set string $s=query(\invec[\indID, \varphi])$ and broadcast $\langle \indID,s,0\rangle$
    \Else
        \State $\indID_L \gets 2\indID-1$;$\indID_R \gets 2\indID$
        \For{$u \in \{\indID_L, \indID_R\}$} in parallel 
            \State Construct the multiset $S(u) \leftarrow \{s~\mid~  \text{received a message of the form }\langle u,s,i-1\rangle\}$ 
            \State $t_{i-1}=2^{i-2} \varphi (\gamma k/n)$
            \State $T_u \gets \ConstDecTree(\FS(S(u),t_{i-1}))$ \label{line:construct}
            \State $s_{u} \gets \Determine(T_u)$
        \EndFor
        \State Set $s_{\indID}$ to $s_{\indID_L}s_{\indID_R}$ and broadcast $\langle \indID,s_{\indID},i\rangle$
        
    \EndIf
\EndFor
\State return the determined string for interval $[1,\dots, n]$
\end{algorithmic}
\end{algorithm}

\paragraph*{Correctness.}

We show correctness by proving the following two key lemmas.

\begin{lemma}\label{lem:high_prob_t_picked}
For every step $i \in [0, \log n]$, every interval in step $i$ (those of size $\varphi_i$) is picked by at least $t_i$ $\text{\nonfaulty}$ {\peers} w.h.p
\end{lemma}
\begin{proof}
Fix Step $i$. The number of intervals in step $i$ is

$\Inum_i$. The expected number of \nonfaulty \peers\ which pick a given interval at Step $i$ is $E_i=\gamma k /\Inum_i=(\gamma k/(n/ (2^i \varphi)) =  2^i (8(c+1) \ln n)$.
Setting $t_i=E_i/2$, the probability of failure for any one interval, as given by Chernoff bounds (see previous subsection) is no more than $e^{-E_i/8}\leq n^{-c+1}$.

Taking a union bound over the $\sum_{i=0}^{\lg\Inum} \Inum_i < \sum_{i=0}^{\infty} \lceil\frac{n}{\varphi}\rceil \frac{1}{2^i} < n$ intervals over all steps $i$, 
the probability of any failure in any step is less than $n \cdot n^{-(c+1)}\leq n^{-c}$ for $c \geq 1$. 
\end{proof}
Denote by $\indID^i_M$ the interval ID picked by \peer $M$ at step $i$.
\begin{lemma} \label{lem:fast_correct_value}
In every step $i\in [0,\lg\Inum]$, every \nonfaulty \peer $M$ learns the correct value of $\invec[\indID^i_M,\varphi_i]$ w.h.p.
\end{lemma}
\begin{proof}
By induction on the steps. The base case, step $0$, is trivial, since every nonfaulty peer that picks an interval queries it completely. 

For step $i>0$, consider the interval $\indID = \indID^i_M$ picked by \peer $M$ in step $i$. During step $i$, $M$ splits $\indID$ into two subintervals $\indID_L, \indID_R$ of size $\varphi_{i-1}$.
By Lemma {\ref{lem:high_prob_t_picked}}, the intervals $\indID_L, \indID_R$ were each picked by a least $t_{i-1}$ \nonfaulty \peers\ 
w.h.p. during step $i-1$, and by the inductive hypothesis, those peers know (and broadcast) the correct strings $s_L, s_R$ respectively.
For $u\in\{L,R\}$, let $\FS_u =\FS(S(\indID_u),t_{i-1})$ 
denote the $t_{i-1}$ frequent sets constructed in step $i$ for $\indID_u$.
Then, $s_L \in \FS_L$ and $s_R \in \FS_R$. Hence, the decision trees built for $\indID_L$ and $\indID_R$ will contain leaves with labels $s_L, s_R$ respectively, which will then be returned correctly by  Procedure \Determine $(\cdot)$, implying that $M$ learns the correct bit string for the interval $\indID$, $s_{\indID} = s_Ls_R$. 
\end{proof}

The correctness of the protocol follows from observing that at the last round $\indID=\lg \Inum$, the intervals are of size $\varphi_\indID = 2^\indID\varphi = n$. Hence, by Lemma {\ref{lem:fast_correct_value}}, every peer learns the entire input.

\paragraph*{Query complexity.}
%
Given an interval $\indID$ in Step $i+1$, let $x$ be the total number of strings received for the subintervals in Step $i$ that compose interval, i.e., $x=x_L +x_R$, where $x_u$ for $u \in \{L,R\}$ denotes the number of strings received for subinterval $\indID_u$.
Let $m_x$ be the number of intervals in Step $i+1$ such that the total number of strings received for that interval equals $x$.
Formally, 
\begin{equation*}
    m_x = |\{\indID \mid \text{received exactly $x$ messages of the form }\langle\indID_L, s, i\rangle \text{ or } \langle\indID_R,s,i\rangle \}|
\end{equation*}
The probability of picking an interval in Step $i+1$ with $x$ strings is $m_x$ divided by the total number of intervals, or 
$m_x/\Inum_{i+1}$. The expected cost of querying in Step $i+1$ is therefore
\begin{equation*}
\sum_x\frac{m_x}{\Inum_{i+1}} \cdot \frac{x}{t_i}
~=~ \sum_x \frac{m_x}{\Inum_{i+1}} \cdot \frac{2x}{\gamma k/\Inum_i}
~=~ \frac{2\Inum_i}{\gamma k \cdot \Inum_{i+1}} \sum_x m_x\cdot x
~\leq~ \frac{4}{\gamma}~,
\end{equation*}
where the last inequality follows since each peer broadcasts at most one string, so $\sum_x m_x \cdot x\leq k$.

Step 0 requires $\varphi=O(n\lg n/(\gamma k))$ queries.
The expected cost of querying is $O(1/\gamma)$ per step $i>0$ per peer, and there are fewer than $\log n$ steps, so the total expected query cost is at most $O(n\log n/(\gamma k))$. 

\paragraph*{Message size.}
The maximum message size is that of the last round or $n/2+1$.

\begin{theorem}\label{thm:fast_expected_optimal}
There is a $O(\log n)$-round algorithm which w.h.p. computes \download\ in the point-to-point model with expected query complexity $\Query=O([n/(\gamma k)]\log n)$ and message size $O(n)$.
\end{theorem}

\subsection{Broadcast Model: Worst-case
\texorpdfstring{$O((1/\gamma)\log^2 n)$}{og} Queries and Time, and Message Size \texorpdfstring{$O(\log n/\gamma)$}{mg}}
\label{ss:opt_in_broadcast_model}

Here we start by showing that in the Broadcast model, 
where every peer (including those controlled by the adversary) must send the same message to all peers in every round, one can drastically reduce the message size in Algorithm \ref{alg:time-query}. Next, we show that the bound on expected query complexity can be improved to a bound on the worst case query complexity because of the ``common knowledge" property guaranteed to the peers by the broadcast medium.

\begin{observation}
\label{obs:simulation}
Any algorithm in the point-to-point model, where in each round, each peer performs at most $Q$ queries and generates at most $r$ random bits, can be simulated in the same number of rounds in the broadcast model, so that in each round of the simulation, each peer also performs at most $Q$ queries, generates at most $r$ random bits, and broadcasts a message of size at most $r + Q$ bits.
\end{observation}
\begin{proof}
In every execution, the state of each node depends on (i) its ID, (ii) its queries, and (iii) the messages it receives. In both the point-to-point model and the broadcast model, all \peers\ know the IDs of the senders. In the broadcast model, all \peers\ receive the same messages. Hence, to communicate its state, each peer needs to send only the random bits it generates and the value of the bits it reads, ordered by their index. That is, any peer receiving this information can generate any longer message the sender would have sent, and thus compute locally the state of every other peer.
\end{proof}

Applying Observation \ref{obs:simulation} 
to the point-to-point algorithm of Section \ref{sec: optimal_time_P2P},
we get the following.

\begin{corollary} 
\label{cor:broadcast-time-query}
There is a $O(\log n)$-round algorithm that w.h.p. performs \download\ in the Broadcast model with expected query complexity and message size $O(n/(\gamma k))$. 
\end{corollary}

\paragraph*{From expected to worst case bounds.} 

Algorithm \ref{alg:time-query} gives only expected bounds on the query complexity, rather than worst case bounds, because the faulty peers may concentrate their negative influence on certain intervals rather than spread themselves out. In the broadcast model, these ``overloaded" intervals are the same for all \peers; therefore, the algorithm can negate this adversarial influence by concentrating the \nonfaulty \peers on these intervals to boost the number of copies of the same string generated by \nonfaulty \peers. 

Consider an iteration $i$ of the for-loop in Algorithm {\ref{alg:time-query}}.
We introduce an additional ``boosting" loop of $j=0$ to $\log n $ iterations to be inserted into each iteration $i$ of Algorithm \ref{alg:time-query}, at its end (after line 11) but before the next iteration. The boosting loop is described by Algorithm \ref{alg:broadcast-time-query}. At the start of the boosting loop, $U$ contains all the intervals of size $2^i \varphi$ and by Lemma \ref{lem:high_prob_t_picked} each interval has been picked by at least $t_i$ nonfaulty peers.  An interval $\indID$ at boosting step $l$ is considered ``overloaded" if the number of strings received for it exceeds $(4/\gamma)2^lt_i$. 
Once a $j$ is reached at which an interval $\indID$ is no longer overloaded, it receives the label $j(\indID)$ and is removed from $U$. At least half of $U$ receives no more than twice the average load and is thus eliminated in each boosting iteration.

At the end of each boosting step $j$, each nonfaulty peer picks a random interval from $U$ to determine and broadcasts it. As the size of $U$ shrinks by a factor of at least 2, the probability of picking any particular interval increases by a factor of $2$ with each increment to $l$, so that after $j$ boosting steps, the number of nonfaulty peers which pick an interval in the remaining $U$ is at least $2^j t_i$.

In order to use the boosting, we need to introduce also a change to line \ref{line:construct} of Algorithm {\ref{alg:time-query}}. The point of boosting the number of copies from \nonfaulty \peers in Iteration $i-1$ is to reduce the cost of determining an interval in the next Iteration $i$ from its two subintervals determined in Iteration $i-1$. 
We modify line \ref{line:construct} by changing the parameter $t_{i-1}$ of \ConstDecTree ~ to $2^{j(\indID)}t_{i-1}$. Note that since we determine intervals from Iteration $i-1$, their labels are already defined in Iteration $i$.

\begin{algorithm} 
\caption{Broadcast algorithm boosting loop subroutine for Iteration $i$}
\label{alg:broadcast-time-query}
\begin{algorithmic}[1]
%
\State $U \leftarrow \{\hbox{intervals of size }2^i\varphi \}$
\For{ $j =0$ to $\log \Inum_i$}  
\For{all intervals $\indID \in U$} let $S(u) \leftarrow \{s~|~ \text{received a message of the form}\langle \indID,s,i\rangle \}$ 
\If{$|\FS(S(\indID), 2^{j} t_i)| \leq 4/\gamma$} 

\State
Remove $\indID$ from $U$
\State $j(\indID) \leftarrow j$ 

\EndIf 

\EndFor
\State Choose a random interval $\indID$ from $U$, {\bf determine} the interval and broadcast it. 
\EndFor
\end{algorithmic}
\end{algorithm}

That is, to {\bf determine} an interval $\indID$ in Iteration $i$ we need to take into account $i$:
\begin{enumerate}
\item Case: {$i=0$:} query the bits in the interval. 
\item Case {$i >0$:} 
Let $\indID_L$ and $\indID_R$ be the subintervals of size $2^{i-1}\varphi$ which comprise $\indID$. \\
For $u \in \{L,R\}$,
construct a decision tree from  $\FS(\indID,2^{j(\indID_u)}t_{i-1})$,
determine each subinterval $\indID_u$ and return their concatenation.
\end{enumerate}

\paragraph*{Correctness.}
We begin with the following lemma:

\begin{lemma}\label{lem:overloaded_halved}
   
    For every Iteration $i$ of Algorithm {\ref{alg:time-query}}, after each iteration $j$ of the boosting loop, $|U| \leq \Inum_i/ 2^{j+1}$.
\end{lemma}
\begin{proof}
    For an interval $\indID$ to remain in $U$ after an iteration $j$ of the boosting loop it is required that $|\FS(S(\indID), 2^{j} t_i)| > 4/\gamma$, which means that more then $(4/\gamma)\cdot 2^{j} t_i$ \peers (faulty and \nonfaulty) picked $\indID$.

    An averaging argument shows that after Iteration $j$, no more than $\Inum_i/(2\cdot 2^j)$ intervals in $U$ can be picked by more than
    $(4/\gamma)2^{j}t_i$ \peers. Hence, the size of $|U|$ is reduced to at most $\Inum_i/2^{j+1}$.
    
    
\end{proof}
By the fact that the boosting loop runs $\log \Inum_i +1$ iterations and as a direct result of Lemma {\ref{lem:overloaded_halved}} we get the following corollary.
\begin{corollary}
    Every interval is assigned a label at some iteration of the boosting loop
\end{corollary}
Next, we show that 
\begin{lemma}
    For every iteration $i$ of Algorithm \ref{alg:time-query},
    every interval that remains in $U$ after the $j$-th iteration of the boosting loop is picked by at least $2^{j+1}t_i$ \nonfaulty \peers w.h.p
\end{lemma}
\begin{proof}
    Setting  $\varphi = (n/\gamma k)(8(c+2)\ln n)$ (changing the value from $\varphi = (n/\gamma k)8(c+1)\ln n$ used in the previous subsection),
    the probability of any interval $\indID \in U$ to be picked after the $j$-th iteration of the boosting loop is $1/|U| \geq 2^{j+1}/\Inum_i$ by Lemma \ref{lem:overloaded_halved}. Hence, the expected number of \nonfaulty \peers that will pick $\indID$ is $E_{i,j} \geq \frac{\gamma k 2^{j+1}}{\Inum_i} = 2\cdot 2^{j+1} t_i$. By Chernoff, the probability that $\indID$ is picked by less then $2^{j+1} t_i$ \nonfaulty \peers is at most $e^{-E_{i,j}/8} \leq n^{-(c+2)}$.

    Taking a union bound over the $\sum^{\log \Inum}_{i=0} \sum^{\log \Inum_i}_{j=0} \Inum_i/2^{j} = \sum^{\log \Inum}_{i=0} \Inum_i \sum^{\log \Inum_i}_{j=0} 1/2^{j} \leq  2\cdot \sum^{\log \Inum}_{i=0} \Inum_i \leq 2n$ intervals over all iterations $i$ and $j$, we get a probability of failure of at most $2n\cdot n^{-(c+2)} \leq n^{-c}$ for $c\geq 1$.
\end{proof}

As a direct result of this lemma we get the following.
\begin{corollary} \label{cor:enough_picked}
If an interval $\indID$ is removed from $U$ during Interval $j$ of the boosting loop then $j(\indID) = j$ and $\indID$ is picked by at least $2^{j(\indID)} t_i$ nonfaulty peers w.h.p.
\end{corollary}

Correctness follows from the observation that by Corollary \ref{cor:enough_picked}, every interval $\indID$ is picked by at least $2^{j(\indID)}t_{i-1}$ during Interval $i-1$ w.h.p. Hence, both subintervals of every picked interval in Iteration $i$ are correctly determined w.h.p, which means the interval is determined correctly w.h.p.


\paragraph*{Query complexity.} 
Determining an interval in Iteration 0, i.e., the first iteration of the for-loop in Algorithm \ref{alg:time-query}, requires $\varphi=O(n\ln n/(\gamma k))$ queries  per peer. Each step of Iteration $0$ has a boosting loop of $\log \Inum$ steps, each of which requires two subintervals of size $\varphi$ to be determined by querying every bit in it, for a total of  $2\varphi$ queries. Each iteration $i>0$ also has a boosting loop and each boosting step requires  that two intervals from the preceding phase be determined, but in this case, the number of queries to determine an interval $\indID$ is equal to $|\FS(S(\indID), 2^{j(\indID)} t_i)|=O(1/\gamma)$. 
Thus the total number of queries per \peer\ is

$O((n/(\gamma k)\log^2 n + (1/\gamma) \log^2 n)=O((n \gamma k) \log^2 n)$.

The {\bf message size} is no greater than the number of random bits needed to select an interval which is $O(\log n)$ plus the number of queries in a step which is no more than $O(n\log n/(\gamma k))$ (in a step of Iteration 0), so that the worst case message size is $O(n\log n/(\gamma k))$.  If we spread each Iteration 0 step over $\log n$ steps, this does not affect the asymptotic running time but it does decrease the message size to 
 $O(\log n +n/(\gamma k))$.

 There are $O(\log n)$ iterations,  each iteration with $i\geq0$ has $O(\log n)$
 rounds for the boosting loop, so the running time is $O(\log^2 n)$.

\begin{theorem} 
\label{thm:broadcast_whp} 
In the Broadcast model, there is a protocol with worst case $O(\log^2 n )$ time, $O((n/\gamma k) \log^2 n)$
 queries, and  $O((n/\gamma k) \log n)$ message size. 
\end{theorem}

\section{Deterministic Download with Crash Faults} \label{sec:crash_faults}

In this section, we present deterministic protocols that solve the \download\ problem under the assumption of a synchronous communication network. First, in Section \ref{sec:static}, we show a warm-up deterministic protocol that achieves optimal query complexity, but its time complexity is $O(nf)$, which is inefficient. Then in Section \ref{sec:rapid}, we show how to improve the time complexity by carefully removing some aspects of synchronizations between \peers.

\subsection{Static Download} 
\label{sec:static}
 
In contrast to the previous section, we consider crash faults instead of byzantine ones. This allows us to circumvent the $\byzfrac n$ lower bound of 
\cite{ABMPRT24}
on the query complexity of \download\, established in \cite{ABMPRT24}. 
We begin by showing a query-optimal deterministic protocol. Later, we present a more complex but also time-optimal protocol.

The static \download\ protocol (Algorithm \ref{alg:static_down}) works as follows.
The execution is partitioned into time segments of $f+1$ rounds each, referred to hereafter as \emph{views}. (In later sections, views may assume other, possibly variable lengths.) Each view $v$ is associated with a publicly known leader denoted $\lead(v)$. We assume throughout that the leader of view $v\geq 0$ is the \peer 
$$
\lead(v) ~=~ v \bmod k.
$$

Each \peer $M$ stores the following local variables.
\begin{itemize}
    \item $\BYZ_M$: The set of crashed \peers that $M$ knows of.
    \item $\view(M)$: $M$'s current view. 
    \item $\indx(M)$: The index of the bit that $M$ currently needs to learn.
   
    \item $\res_M$: $M$'s output. We write $\res_M[i] = \bot$ to indicate that $M$ did not yet assign 
    a value to the $i$'th cell. 
   
\end{itemize}

\paragraph{View $v$ structure.}
First, the leader of the view, $\lead(v)$, queries $b_i$ where $i=\indx(\lead(v))$ (we show in Lemma \ref{lem:in_sync} that at any given time, $\indx(M)$ is the same for every \peer $M$). 
Then, for $f+1$ rounds, every \peer $M$ sends $\res_M[\indx(M)]$ (if it knows 

it, i.e., $\res_M[\indx(M)]\neq \bot$) to every other \peer.
After $f+1$ rounds, every \peer $M$ checks its local value $\res_M[\indx(M)]$. If $\res_M[\indx(M)] = \bot$, it concludes that $\lead(v)$ has crashed and adds it to $\BYZ_M$ (and the bit $\indx(M)$ will have to be queried again in the next view, by another leader). Otherwise, it increases $\indx(M)$ by 1. (We show in Lemma \ref{lem:local_sync} that either every \peer $M$ has $\res_M[\indx(M)] = \bot$ or every \peer $M$ has $\res_M[\indx(M)]\in \{0,1\}$.)
Finally, every \peer $M$ enters the next view $v'$ such that $\lead(v') \notin \BYZ_M$.
We remark that once a new bit is learned, and the {\peers} enter the next view, the leader is replaced even if it did not crash. This is done in order to balance the workload over the {\peers}.

See Algorithm \ref{alg:static_down} for the formal code.

\begin{algorithm}
    \caption{Static \download\ (Code for \peer $M$)} 
    \label{alg:static_down}
    \begin{algorithmic}[1]
    \State \textbf{Local variables} 
        \State $\res \gets \emptyset$
        \State $\BYZ \gets \emptyset$
        \State $\indx(M) \gets 1$
        \State $\view(M) \gets 1$
        \State
        
        \While{$\indx(M) \leq n$}
            \If{$\lead(v) = M$} \Comment{Query Step}
                    \State $\res[\indx(M)] \gets \DataQuery(\indx(M))$
                \EndIf
                \For{$t=1, \dots, f +1$} \Comment{Message Step}
                     \If{$\res[\indx(M)] \ne \bot$}
                    Send a \view $v$ message
                    $\langle \view\ v, \indx(M), \res[\indx(M)] \rangle$
                    to all \peers.
                    \EndIf
                    \State Receive \view $v$ messages from all other \peers.
                    \If{received a \view $v$ message containing 
                    $\langle i, b\rangle$ from at least one \peer}
                        \State 
                        $\res[i] \gets b$
                    \EndIf
                \EndFor
                \If{$\res[\indx(M)] = \bot$} \Comment{View change Step}
                    \State $\BYZ \gets \BYZ \cup \{\lead(v)\}$
                \Else
                    \State $\indx(M) \gets \indx(M) + 1$
                \EndIf
                \State Let $v' >v$ be the least view such that $\lead(v') \notin \BYZ$ and set $\view(M) \gets v'$
                \label{line:next_view}
        \EndWhile
    \end{algorithmic}
\end{algorithm}

Since every view takes exactly $f+1$ rounds, we get the following.
\begin{observation}\label{obs:same_pace}
    Every \peer $M$ is at the same step/round within its current view, $\view(M)$.
\end{observation}

For every view $v$, let $t^s_v$ be the first round a \peer $M$ set $\view(M)=v$ and $t^e_v$ be the first round where a \peer $M$ set $\view(M)=v'$ for $v'>v$. We say that the \peers are \emph{in sync in view $v$} 
if at round $t^s_v$, $\view(M)=v$ and $\indx(M)$ and $\BYZ_M$ are the same for every \peer $M$, and at round $t^e_v$, every \peer enters some view $v'>v$ (it follows from the ensuing analysis that it is the same $v'$ for every \peer). 

When the \peers are in sync in view $v$, we denote by $\indx(v)$ the common value of $\indx(M)$ and by $\BYZ_v$ the common value of $\BYZ_M$ for every \peer $M$.

\begin{lemma} \label{lem:local_sync}
    If the \peers are in sync in view $v$, then after the message step in view $v$, the bit $\indx(v)$ is either downloaded (and known to all {\peers}) or missing (i.e., not known to any of the {\peers}).
\end{lemma}
\begin{proof}
    If every \nonfaulty \peer has $\res[\indx(v)] = \bot$ at the end of view $v$, then $\indx(v)$ is missing. Otherwise, let $M$ be a \nonfaulty \peer with $\res[\indx(v)] = b$ at the end of view $v$. Consider the following cases.

    \inline Case 1 - $M$ first received $b$ at iteration $t<f+1$:
        At the next iteration $t+1$, $M$ sends $\langle \indx(v), b \rangle$ 
        to every other \peer, and since it is not faulty, $\res[\indx(v)] = b$ for every \nonfaulty \peer, so $\indx(v)$ is downloaded after the for loop.
    
    \inline Case 2 - $M$ first received $b$ at iteration $t=f+1$:
        In this case, there exists a chain of \peers $\hat{M}_1, \dots, \hat{M}_{f+1}$ such that $\hat{M}_j$ sends $b$ to $\hat{M}_{j+1}$ on rounds $1 \leq j\leq f$ of the for loop in the message step and $\hat{M}_{f+1}$ sends $b$ to $M$. Since there are $f+1$ \peers in the chain (no duplicates are possible since it would indicate that a crashed \peer sent a message after it crashed), at least one of them is \nonfaulty. Let $\hat{M}$ be the \nonfaulty \peer in the chain. We observe that $\hat{M}$ fulfills the condition of Case 1, so $\indx(v)$ is downloaded after the for loop.
\end{proof}

\begin{corollary}\label{cor:same_local_values}
    If \peers are in sync in view $v$, then they share a common value of their local variables throughout the view.
\end{corollary}
\begin{proof}
    The statement is true by definition at the beginning of the view. The only time the local variables change after the start of the view is during the view change step. At that time, by Lemma \ref{lem:local_sync},
    $\indx(v)$ is either downloaded, in which case every \peer $M$ sets $\indx(M)= \indx(v)+1$, or it is missing, in which case every \peer $M$ adds $\lead(v)$ to $\BYZ_M$. Overall, $\BYZ_M = \BYZ_{M'}$ and $\indx(M)= \indx(M')$ for every pair of \peers $M, M'$.
\end{proof}

\begin{lemma} \label{lem:in_sync}
At every round, the \peers are in sync in some view $v\geq 1$.
\end{lemma}

\begin{proof}
Towards contradiction, let $v$ be the first view where the \peers are not in sync in $v$. Hence, on round $t^s_v$ either there exists a \peer $M$ that didn't enter view $v$ or not all \peers have the same values in their local variables (note that for $v=1$, this cannot happen by the initialization values of Algorithm {\ref{alg:static_down}} so, $v\geq 2$).
    
The former cannot happen because by Observation \ref{obs:same_pace}, every \peer enters a new view at the same time, so $M$ must have started some view $\hat{v} \neq v$ at round $t^s_v$. Yet, the \peers were in sync in every view $v'<v$ which means that by Corollary \ref{cor:same_local_values}, in the view prior to $v$, say $v'$, every \peer $M$ had the same value of $\BYZ_M $ meaning that every \peer enters the same view at the end of view $v'$, so $\hat{v} = v$, contradiction. 
    
The latter cannot happen because the local values of \peer $M$ at round $t^s_v$ are the same as the local values at round $t^e_{v'}$ of \peer $M$, where $v'<v$ is the view prior to $v$, for every \peer $M$, and by Corollary \ref{cor:same_local_values} those are the same for every \peer.
\end{proof}

\begin{lemma}
    If the \peers are in sync in view $v$, then when a \peer sets $\res[i] \gets b$, it follows that $b=b_i$.
\end{lemma}
\begin{proof}
    For \peer $M$ to set $\res[i]\gets b$, there must exists a sequence of \peers \- $M_1, M_2, \dots, M_r$ of length $r$ such that $M_r=M$, $M_j$ sent $b$ to $M_{j+1}$ for $j=1,\dots, r-1$ and $M_1$ made a query and received $b$. Since the \peers are in sync, the local $\indx(M)$ of every \peer $M$ is the same and equals $i$, so $M_1$ queried $b_i$, hence $b=b_i$.
    \end{proof}

It follows that if $\res[i] \neq \bot$, then $\res[i] = b_i$.

Note that $\indx(v)\leq \indx(v+1) \leq \indx(v)+1$ and $\indx(1)=1$. 
We get the following corollary by combining Lemmas \ref{lem:local_sync} and \ref{lem:in_sync}.
\begin{corollary} \label{cor:downloaded_or_missing}
    After every view $v$, $\indx(v)$ is either downloaded or missing.
\end{corollary}

\begin{lemma} \label{lem:nonfaulty_leader}
    For view $v>0$, if $\lead(v)$ is \nonfaulty, then $\indx(v)$ is downloaded at the end of $v$.
\end{lemma} 

\begin{proof}
    At the beginning of view $v$, $\lead(v)$ queries the bit $\indx(v)$ and sends it to every other \peer. Subsequently, every \nonfaulty \peer sets $\res[\indx(v)] \gets b_{\indx(v)}$ so $\indx(v)$ is downloaded.
\end{proof}

By Lemma \ref{lem:nonfaulty_leader} and Corollary \ref{cor:downloaded_or_missing}, we get the following corollary.

\begin{corollary} \label{cor:justified_blacklist}
    For view $v>0$, if $\indx(v)$ is missing at the end of view $v$, then $\lead(v)$ crashed, and it is added to $\BYZ$.
\end{corollary}

\inline Correctness: 
By Corollary \ref{cor:downloaded_or_missing}, after every view $v$, $\indx(v)$ is either downloaded or missing. By Corollary \ref{cor:justified_blacklist}, if $\indx(v)$ is missing, $\lead(v)$ will be added to $\BYZ$ and will never be the leader again. Hence, there are at most $f$ views $v$ at the end of which $\indx(v)$ is missing. at the end. So, after $n+f$ views, all $n$ bits are downloaded.

\inline Complexity:
Call a view \emph{good} (respectively, \emph{bad}) if $\indx(v)$ is downloaded (resp., missing) by the end of it. 
As explained above, the protocol is finished after at most $n+f$ views, $n$ of which are good, and at most $f$ are bad. Each view $v$ incurs one query and takes $O(f)$ rounds. 

By Lemma \ref{lem:nonfaulty_leader}, the good views are led by \nonfaulty \peers in a round-robin fashion. Hence, every \nonfaulty \peer leads at most $\frac{n}{\goodfrac k}$ good views. By Corollary \ref{cor:justified_blacklist},  the bad views are led by crashed \peers, and each such \peer may be a leader of a bad view at most once (since it is subsequently added to $\BYZ$). Generally, a \peer may be a leader at most $\frac{n}{\goodfrac k}$ times.
Hence, the query complexity is $\Query=O(\frac{n}{\goodfrac k})$, and the time complexity is $\Time= O(nf)$. Moreover, since at every iteration of the for loop every {\peer} sends a message, we get a message complexity of $\Message = O(k^2\cdot f\cdot n)$.

\begin{theorem} \label{thm:sync_static_download}
In the synchronous model with $f\le \byzfrac k$ crash faults where $\byzfrac<1$, there is a deterministic protocol that solves \download\ with $\Query = O(\frac{n}{\goodfrac k})$, $\Time = O(nf)$ and $\Message = O(k^2\cdot f n)$.
\end{theorem}


\subsection{Rapid Download } 
\label{sec:rapid}

Algorithm \ref{alg:static_down} has optimal query complexity and is simple but is not optimal in time complexity.
It uses a standard technique of sending a value for $f+1$ rounds, ensuring that by the end, either every {\peer} commits the value or none of them do.
This creates a bottleneck that prolongs the execution time of the protocol.
It is worth noting that it was shown in {\cite{FLP}} that in the presence of $f$ crash faults, agreement must take at least $f+1$ rounds. Hence, shortening the length of a view inevitably breaks the property described above.
We present the following improved algorithm to overcome this bottleneck.
%
%

Like Algorithm \ref{alg:static_down}, the algorithm is partitioned into views 
and uses the notations $\view$, $\lead$, $\BYZ$, and $\indx$ as defined earlier. 
Unlike that algorithm, in which views were changed automatically on fixed rounds, here the {\peers} need to initiate a \say{view change} process, and send \say{view change} messages in order to change the view. 
This allows us to shorten the span of every view from $f+1$ rounds to just two rounds, but it also brings about some difficulties in coordination. Specifically, Algorithm {\ref{alg:static_down}} ensures that at any given time, every {\peer} has the same values in its local variables as every other {\peer}, while in Algorithm {\ref{alg:rapid_down}}, different {\peers} can be in different views at the same time, with a different set of known bits and a different set of known crashed {\peers}. 
Because of that, it is necessary to carefully construct the transition process between views, referred to as \emph{view change}, so that the most and least advanced {\peers} are at most one bit index apart (in terms of their local variable {$\indx(M)$}) and that progress is always guaranteed.

A view $v$ starts after the leader of that view, $\lead(v)$, first receives messages of type \say{view change} at some round $t$.

We proceed with a more detailed description of the algorithm.

\paragraph{Leader instructions in view $v$.}
The \peer $M$ becomes a leader when it receives a \viewchange $v$ message for the first time. It then starts executing the leader instructions described above.

When view $v$ starts at the beginning of some round $t$ (the first round when $\lead(v)$ received \viewchange $v$ messages), $\lead(v)$ picks the \viewchange $v$ message with the highest index $i$ amongst those received at the end of round $t-1$ and queries the bit $b_i$. Then, $\lead(v)$ sends a \view $v$ message$\langle \view, v, i, b_i \rangle$ to every other \peer, \emph{twice}, in two consecutive rounds.

\paragraph{Instructions when receiving a \view $v$ message.}
When a non-leader \peer $M$ in view $v$ receives a \view $v$ message $\langle \view, v, i, b \rangle$, it first updates its local $\view(M)$ variable to $v$ if $v\geq \view(M)$. Then, it updates its local $\indx(M)$ variable to $i$ if $i\geq \indx(M)$. Finally, it stores $b$ in $\res_M[i]$. If the message is the second one received from $\lead(v)$, then $M$ increases $\indx(M)$ by 1 and moves to the next `available' view 
$v'$ (i.e., such that $\lead(v')$ is not in the list $\BYZ_M$ of crashed \peers),
by a view change process detailed below.

If a non-leader \peer $M$ has $\view(M)=v$ (and is not during a view change) and didn't receive a \view $v$ message from $\lead(v)$, it adds $\lead(v)$ to $\BYZ_M$ and initiates a view change process as before.

\paragraph{View change instructions.}
When a \peer $M$ initiates a \viewchange $v$ process at time $t$, it invokes Procedure \viewchange in which it sends the next leader, $\lead(v)$, a \viewchange $v$ message 
$\langle \viewchange\ v, \indx(M) \rangle$.
Subsequently, at time $t+1$, $M$ enters an Idle state for one time step (until time $t+2$) while it waits for $\lead(v)$ to receive the previously sent \viewchange $v$ message, in which it does not expect to receive {\view} messages from any leader (during that time step, $\lead(v)$ collects \viewchange $v$ messages and starts sending {\view} messages that will be received at time $t+2$). Note that if a {\view} $v'\geq v$ message is received during the Idle state (at time $t+1$), then $M$ leaves the Idle state early and proceeds as described above
(this may happen if some other {\peer} $M'$ has initiated a view change one round earlier than $M$). Hence, at the end of a \viewchange $v$, $M$ enters view $v$ only if $v > \view(M)$.

\paragraph{Remarks.}
We make two remarks about the possible behavior of the algorithm. First, note that it is possible for $M=\lead(v)$ to get {\viewchange} messages in more than one round. In particular, suppose that in some view $v'$, $\lead(v')$ is faulty, and in its first round (time $t$) it informs the {\peer} $M'$ but does not inform $M''$. At time $t+1$, $M''$ marks $\lead(v')$ faulty and invokes \viewchange, so the next leader, $M=\lead(v)$, gets its message at time $t+2$. At the same time ($t+2$), $M'$ realizes that $\lead(v')$ has failed, and invokes \viewchange, so $M$ gets its message at time $t+3$. Note that this example also shows that a view might be of different lengths for different \peers.

Second, note that it is possible that $M=\lead(v)$ will receive {\viewchange} messages with different indices $i,i+1$, in consecutive rounds. A possible scenario where this might happen is as follows. Suppose the previous leader, $M'=\lead(v')$, failed in its second transmission round, sending the {\view} message to $M''$ but not to $M'''$. then $M''$, having received the message twice, invokes {\viewchange} $v$ with index $\indx(M'')+1$, while $M'''$, having missed the second message, invokes {\viewchange} $v$ with index $\indx(M''')$. Since $\indx(M'') = \indx(M''')=i$, $M$ will get {\viewchange} messages with $i$ and $i+1$. In this case, $M$ will work on $i+1$, and the bit $b_i$ will never be transmitted again. This is fine, however, since this bit has already been sent to everyone by the previous leader, $M'$, on its first transmission round.

\commentstart
%
%
Figures {\ref{fig:execution_timeline_nofaults}} and {\ref{fig:execution_timeline_with_faults}} illustrate two example executions. The first involves \nonfaulty leaders and shows the ``normal" way the execution should look like. The second involves a faulty leader, crashing on its first transmission round, and illustrates how {\peers} change views adaptively and handle crashes.

%
%
\def\FIGURES{
\begin{figure}
    \centering
    \includegraphics[width=0.85\linewidth]{execution_timeline_no-faults.jpeg}
    \caption{Execution timeline with \nonfaulty leaders}
    \label{fig:execution_timeline_nofaults}
\end{figure}

\begin{figure}
    \centering
    \includegraphics[width=0.85\linewidth]{execution_timeline_with faults.jpeg}
    \caption{Execution timeline with faulty leaders}
    \label{fig:execution_timeline_with_faults}
\end{figure}
} 
\FIGURES
\commentend

See Algorithm {\ref{alg:rapid_down}} for the formal code. 

\begin{algorithm}[htb]
    \caption{Rapid \download\ (Code for \peer $M$)}
    \label{alg:rapid_down}
    \small
    \begin{algorithmic}[1]
    \State \textbf{Local variables} 
        \State $\BYZ$, initially $\emptyset$ \Comment{The set of crashed \peers $M$ recognized}
        \State $\view(M)$, initially 1 \Comment{This is the current view of $M$}
        \State $\indx(M)$, initially 1 \Comment{The index of the current bit $M$ needs to learn}
        \State $\res[i]$, $i\in \{1,\dots, n\}$, initially $\bot$ \Comment{Output array}
        
                
        
        \State \Comment{Instructions for view $v$ leader}
        \Upon{receiving \viewchange $v$ message for the first time 
        }
        
        \Comment{View changes to $v$, $M=\lead(v)$ becomes leader}
        
        \State From all messages $\langle \viewchange, v, i\rangle$ received in the current round, select  the highest  index $i$.
        \If{$\indx(M) < i$}

\State $\indx(M) \gets i$
        \EndIf
        \If{$\res[\indx(M)] = \bot$}
                \State $\res[\indx(M)] \gets \DataQuery(\indx(M))$ 
            \EndIf
        \For{$j=1,2$}
            
            \State Send a \view $v$ message 
            $\langle \view\ v, \indx(M), \res[\indx(M)] \rangle$ to every \peer.
        \EndFor
        \EndUpon
        \State \Comment{\peer instructions in view $v$}
        \Upon{receiving a  message 
        $\langle \view~ v, i, b_i \rangle$}
    
            \If{$v \geq \view(M)$}
                
                \State $\view(M)\gets v$
                \If{$i \geq \indx(M)$}
                    \State$\indx(M) \gets i$
                \EndIf
                \State $\res[i] \gets b$
            \EndIf

            \If{this is the second \view $v$ message received}
                \State $\indx(M)\gets \indx(M)+1$
            \State Invoke Procedure $\viewchange$.
            \EndIf
        \EndUpon
        \State
        \Upon{not receiving a message from 
        $\lead(\view(M))$ for an entire round 
        and 
        not idle}
        
            \State 
            $\BYZ \gets \BYZ \cup \{\lead(\view(M))\}$

            \State Invoke Procedure $\viewchange$.
        \EndUpon

        \State \Comment{view change instructions}

        \Procedure{\viewchange}{}
            \State Let $v' >\view(M)$ be the least view such that $\lead(v') \notin \BYZ$.
            \State Send a {\viewchange} message 
            $\langle \viewchange\ v', {\indx(M)} \rangle$ to $\lead(v')$
            \State Stay idle for one round. 
            \State \textbf{if} $v\geq \view(M)$ \textbf{then} $\view(M) \gets v'$.
        \EndProcedure
        
    \end{algorithmic}
\end{algorithm}

\paragraph*{Analysis.}
We now establish correctness and analyze the algorithm's complexity.

We say that a view $v$ is \emph{\nonfaulty} if $\lead(v)$ is \nonfaulty during view $v$ and \emph{faulty} otherwise.

\begin{lemma}\label{lem:leader_active}
    If, at any time $t$, some \peer $M$ has $\view(M)=v$, then either $\lead(v)$ also has $\view(\lead(v)) = v$, or $\lead(v)$ has crashed.
\end{lemma}
\begin{proof}
A \peer $M$ enters view $v$ either by performing a view change or by getting a \view $v$ message from $\lead(v)$. In the former case, $\lead(v)$ has not yet entered view $v$. After the view change, both $M$ and $\lead(v)$ enter view $v$ and by the structure of the view they both remain in view $v$ for one or two more rounds, until either $\lead(v)$ crashes or two \view $v$ messages are received, which results in the end of view $v$. In the latter case, $\lead(v)$ has already entered view $v$ and $M$ receives the first \view $v$ message (out of two) and enters view $v$. Both $M$ and $\lead(v)$ remain in view $v$ for one more round, until the second message is received or $\lead(v)$ crashes.
\end{proof}

\begin{lemma} \label{lem:BYZ_justified}
If a \peer $M$ adds $\lead(v)$ to $\BYZ_M$, then $\lead(v)$ has crashed.
\end{lemma}
\begin{proof}
A \peer $M$ adds $\lead(v)$ to $\BYZ_M$ if'f $\view(M)=v$ and $M$ didn't receive a \view $v$ message from $\lead(v)$ (at some round other than the idle round).

By Lemma \ref{lem:leader_active}, it cannot be the case that $\lead(v)$ is in a different view, and since it hasn't sent a \view $v$ message while in view $v$, it must have crashed.
\end{proof}
\begin{lemma} \label{lem:sync_view}
    At the beginning of the second round of a \nonfaulty view $v$, every \peer $M$ has $\view(M)= v$.
\end{lemma}
\begin{proof}
During the first round of a \nonfaulty view $v$, $\lead(v)$ sends a \view $v$ message to every \peer. Subsequently, every \peer that is in a previous view advances to view $v$ (while every other \peer remains unaffected). Therefore, after the first round of view $v$, $\view(M)\geq v$ for every \peer $M$. Assume towards contradiction that there exists some \peer $M'$ such that $\view(M') > v$ at the beginning of the second round of view $v$. For that to be the case, $M'$ must have added $\lead(v)$ to $\BYZ_{M'}$, in contradiction to Lemma \ref{lem:BYZ_justified} and the fact that view $v$ is \nonfaulty.
\end{proof}

\begin{lemma}\label{lem:confirmed_bit}
For every two nonfaulty \peers $M, M'$ and index $i$, if $i < \indx(M)$ then $\res_{M'}[i] = b_i$. 
\end{lemma}
\begin{proof}
Let $M$ be a \nonfaulty \peer. 
For every $i < \indx(M)$, let $M_i$ be the first \nonfaulty {\peer} to set $\indx(M_i) > i$. Since $M_i$ is the first to increase $\indx(M_i)$ above $i$, it must be the case that 
$M_i$ received two \view $v$ messages from some leader, 
$\langle \view~ v, i, b_i \rangle$.
Since 
$M_i$ received two messages, every \nonfaulty \peer $M'$ must have received at least one such message, and subsequently set $\res_{M'}[i]=b_i$.
\end{proof}

Let $\minbit(t) = \min_{\mbox{\footnotesize \peers~} M}\{\indx(M)~\mbox{at round}~t\}$. 
\begin{lemma}\label{lem:min_bit_increase}
For every view $v\geq 1$, $\minbit(t^e_v)\geq \minbit(t^s_v) +1$.
\end{lemma}
\begin{proof}
By Lemma \ref{lem:sync_view}, at the beginning of the second round of view $v$ every \peer $M$ has $\view(M)= v$. Hence, when view $v$ finishes, at round $t^e_v$, every \peer $M$ has $\indx(M)\geq \indx(\lead(v))+1 \geq \minbit(t^s_v)+1$. Hence, the Lemma follows.
\end{proof}
\begin{lemma} \label{lem:time_fast_rapid}
    After $n+f$ views, every \peer knows the entire input $\invec$
\end{lemma}
\begin{proof}
    By Lemma \ref{lem:min_bit_increase}, $\minbit$ increases after every \nonfaulty view. Out of $n+f$ views, at most $f$ might be faulty, by Lemma \ref{lem:BYZ_justified} and the fact that once a \peer $M$ adds another \peer $M'$ to $\BYZ_M$ it will not enter any view $v$ where $\lead(v) = M'$. Therefore, after $n+f$ views, $\indx(M) \geq n+1$ for every \peer $M$. By Lemma \ref{lem:confirmed_bit}, this means that $\res_M[i] = b_i$ for every $1\leq i\leq n$ and \peer $M$.
\end{proof}
\begin{theorem} \label{thm:sync_rapid_download}
In the synchronous  model with $f\le \byzfrac k$ crash faults where $\byzfrac<1$, there is a deterministic protocol that solves \download\ with $\Query = O(\frac{n}{\goodfrac k})$, $\Time = O(n+f)$ and $\Message = O\left(k\cdot (n+f)\right)$.
\end{theorem}
\begin{proof}
By Lemma \ref{lem:time_fast_rapid}, after $n+f$ views, every \peer knows the entire input. A view takes at most two rounds; hence, the time complexity is $\Time=O(n+f)$. By Lemma \ref{lem:min_bit_increase}, every \nonfaulty view increases $\minbit$ by at least 1, so after $n$ \nonfaulty views the protocol terminates. Because we use round robin order for leaders, every \nonfaulty \peer is the leader of at most $n/\goodfrac k$ views. Hence, the query complexity is $\Query=\frac{n}{\goodfrac k}$. Moreover, since every non leader {\peer} communicates only with its current leader and every leader communicates with every other {\peer}, the number of messages per round is $O(k)$, so
$\Message = O(k(n+f))$.
\end{proof}
\section{The Asynchronous Data Retrieval Model with Crash Faults} \label{sec:async_crash_faults}

\subsection{\download\ with At Most One Crash}

We start the exploration of the asynchronous setting with an algorithm that solves \download with at most one crash fault. This serves as an introduction to our main algorithm for handling an arbitrary number of crashes.

The algorithm runs in two phases. In each phase, every \peer maintains a list of assigned indices. Each phase has three stages, and every message contains the local phase number and stage number. We describe a single phase and particularly the operation of \peer $M_i$.

In stage 1, the \peer $M_i$ queries all assigned bits that are still unknown and sends stage-1 messages containing the assigned bits' values to all other \peers. (Assuming it is necessary to send the assigned bits in multiple packets, each packet would also include the packet number.) 

In stage 2, every \peer $M_i$ waits until it receives all stage-1 messages (according to its local assignment) from at least $k-1$ \peers (waiting for the last {\peer} risks deadlock, in case that {\peer} has crashed). When that condition is met, $M_i$ sends a stage-2 message containing the index $j$ of the \emph{missing} \peer, namely, the \peer $M_j$ from which it didn't receive all stage-1 messages during the current phase. When $M_i$ receives a stage-2 message containing the index $j$, it sends a stage-2 response containing either the bits assigned to $M_j$ (if it heard from $M_j$) or \metoo\ if it didn't hear from $M_j$ during stage 2. (In case $M_i$ hasn't finished waiting for stage-1 messages, it delays its response until it is finished.)

Finally, in stage 3, every \peer $M_i$ waits until it collects at least $k-1$ stage-2 responses. When that happens, there are two cases. Either $M_i$ has received only \metoo\ messages, in which case it reassigns the bits of $M_j$ evenly to all \peers and starts the next phase, or $M_i$ has received $M_j$'s bits from at least one \peer, in which case it goes into \textit{completion} mode, which means that in the next phase, $M_i$ acts as follows. In stage-1, it sends all the bits in stage-1 messages. In stage-2, it doesn't send stage-2 messages, and in stage-3, it doesn't wait for stage-2 responses.

After two phases, every \peer terminates.

Note that if a \peer receives a message from (another {\peer} which is in) a later phase, it stores it for future use, and if it receives a message from an earlier phase, it retrieves the bits the message (possibly) holds and evaluates whether or not to enter completion mode 

See Algorithm \ref{alg:single_crash_async_down} for the formal code.

\begin{algorithm}
    \caption{Asynchronous Download one crash (code for \peer $M_i$)} \label{alg:single_crash_async_down}
    \begin{algorithmic}[1]
    \State $Mode \gets$ Active
    \For{$t \in \{1, 2\}$}
        \Procedure{Stage 1}{}
        \State Query all unknown assigned bits.
        \If{$Mode =$ Active}
            \State Send a stage-1 message with the assigned bits to every other \peer. 
        \Else
            \State Send a stage-1 message with all known bits.

        \EndIf

        \State (In the CONGEST model, this message may need to be broken into packets of size $O(\log n)$.)
    \EndProcedure
    \Procedure{Stage 2}{}
        \State Maintain the following set:
        \Statex $H = \{j \mid \text{received all stage-1 messages from } M_j \}$
        \State When $|H| \geq k-1$, if $Mode=$Active, send $J_f(i)$ s.t. $\{J_f(i)\} = \{1, \dots, k\} \setminus H$ to all other \peers. 
        \State  \textbf{When}
        receiving a stage-2 message $J_f(i')$ from \peer $M_{i'}$:
        \Statex Send back a stage-2 response containing \metoo\ if $J_f(i') = J_f(i)$ and $M_{J_f(i')}$'s assigned bits otherwise. 
    \EndProcedure
    \Procedure{Stage 3}{}
        \If{$Mode=$Active}
        \State Collect stage-2 responses until receiving at least $k-1$ messages.
        \If{Received only \metoo\ responses}
            \State Reassign $M_{J_f(i)}$'s bits to \peers $I =\{1, \dots, k \} \setminus \{J_f(i)\}$, and start the next phase. 
        \Else
            \State $Mode \gets$ completion
        \EndIf
        \EndIf
    \EndProcedure
    \Procedure{messages from different phases}{}
        \If{
        Received a message from a different phase}
            \State If it is from a later phase, store it for future use.
            \Statex If it is from a previous phase, evaluate its content and update your known bits accordingly. If you have no more unknown bits, go into completion mode. 
        \EndIf
    \EndProcedure
    \EndFor
    \end{algorithmic}
\end{algorithm}

We use the following facts to show the correctness and compelxity of the protocol.

\begin{observation} \label{obs:overlap}
    (Overlap Lemma) Assuming $2f<k$, every two sets of $k-f$ \peers must overlap at least one \peer.
\end{observation}
\begin{proof}
Exactly $f$ distinct \peers are not present in one of the sets, hence at least $k-2f \geq 1$ \peers in the other set must not be distinct, thereby being in both sets.
\end{proof}

The following observation holds after each of the two phases, although a stronger property holds after phase 2, namely, each \nonfaulty \peer knows all input bits.

\begin{observation} \label{obs:peers_missing}
After Stage 2, every \peer lacks bits from at most one ``missing" \peer.
\end{observation}

This observation is obvious from the algorithm since each \peer receives all stage-1 messages from at least $k-1$ \peers.

\begin{lemma} \label{lem:stage_3_uniquness}
After stage 3, if two different \peers $M_i$ and $M_{i'}$ lack bits from missing \peers $M_j$ and $M_{j'}$ respectively, then $j=j'$.
\end{lemma}

\begin{proof}
Let $M_i$ and $M_{i'}$ be \peers as required by the premises of the lemma. $M_i$ received $k-1$ \metoo\ messages for $j$ and $M_{i'}$ received $k-1$ \metoo\ messages for $j'$. By the Overlap Lemma \ref{obs:overlap}, there is at least one \peer that sent \metoo\ to both $M_i$ and $M_{i'}$ for $j$ and $j'$ respectively, yet, by Observation \ref{obs:peers_missing}, each \peer lacks bits from at most one missing \peer. Hence, $j=j'$.
\end{proof}

\begin{theorem} \label{thm:async_single_crash}
In the asynchronous model with $1$ crash fault, there is a deterministic protocol that solves Download with $\Query =\frac{n}{k} + \lceil\frac{n}{k(k-1)}\rceil$, $\Time = \tilde{O}\left(\frac{n}{k}\right)$ and $\Message = O(nk)$.

\end{theorem}
\begin{proof}
By Observation \ref{obs:peers_missing}, after stage 2 of phase 1, each \peer $M_i$ lacks bits from at most one missing \peer. By Lemma \ref{lem:stage_3_uniquness}, after stage 3 of phase 1, every \peer $M_i$ that still lacks some bits has the same missing \peer $M_j$. Consider such a \peer $M_i$. At the end of stage 3 of phase 1, $M_i$ reassigns the bits of $M_j$ evenly to $I =\{1, \dots, k\} \setminus \{j\}$. Every \peer in $I$ either lacks some bits from $M_j$ or is in completion mode, so in the following phase 2, they will send stage-1 messages consistent with the local assignment of $M_i$. Hence, in stage 2 of phase 2, either $M_i$ receives a phase-2 stage-1 message from $M_j$, meaning it also receives a phase-1 stage-1 message from $M_j$, or it receives all phase-2 stage-1 messages from all \peers in $I$. In both cases, there are no unknown bits after stage 2 of phase 2.

In terms of query complexity, each \peer queries $n/k$ times in phase 1. Because each \peer lacks bits from at most one \peer, it has at most $n/k$ unknown bits. At the end of stage 3 of phase 1, every such \peer reassigns those $n/k$ bits evenly to the remaining $k-1$ \peers, resulting in additional $n/(k(k-1))$ queries per \peer. In total, the query complexity is $n/k + n/(k(k-1))$.
\end{proof}

\subsection{Extending the Result to \texorpdfstring{$f$}{} Crashes}
    In this subsection, we present a protocol that extends Algorithm \ref{alg:single_crash_async_down} to an algorithm that can tolerate up to $f$ crashes for any $f<k$.
    
    The main difficulty in achieving tolerance with up to $f$ crashes is
    that in the presence of asynchrony, one cannot distinguish between a slow \peer and a crashed \peer,
    making it difficult to coordinate.
    
    Similarly to Algorithm \ref{alg:single_crash_async_down}, Algorithm \ref{alg:f_crash_async_down} executes in phases, each consisting of three stages. Each \peer $M$ stores the following local variables.
    \begin{itemize}
        \item $phase(M)$: $M$'s current phase.
        \item $stage(M)$: $M$'s current stage within the phase.
        \item $H_p^M$: the \emph{correct set} of $M$ for phase $p$, i.e., the set of \peers $M$ heard from during phase $p$.
        \item $\sigma_p^M$: the \emph{assignment function} of $M$ for phase $p$, which assigns the responsibility for querying each bit $i$ to some \peer $M'$.
        \item $\res^M$: the output array.
    \end{itemize}
    We omit the superscript $M$ when it is clear from the context.
    
    In the first stage of phase $p$, each \peer $M$ queries bits according to its local assignment $\sigma_p$ and sends a $\phase \ p \ \stage \ 1$ request (asking for bit values according to $\sigma_p$, namely $\{i \mid \sigma_p(i) = M'\}$) to every other \peer $M'$ and then continues to stage 2.
    Upon receiving a $\phase \ p \ \stage \ 1$ request, $M$ waits until it is at least in stage 2 of phase $p$ and returns the requested bit values that it knows.
    In stage 2 of phase $p$, $M$ waits until it hears from at least $|H_p^M| \geq k-f$ \peers (again, waiting for the remaining $f$ {\peers} risks deadlock). Then, it sends a $\phase \ p \ \stage \ 2$ request containing the set of \peer Id's $F_p^M = \{1, \dots, k\}\setminus H_p^M$ (namely, all the \peers it didn't hear from during phase $p$) and continues to stage 3.
    Upon receiving a $\phase \ p \ \stage \ 2$ request, $M$ waits until it is at least in stage 3 of phase $p$, and replies to every \peer $M'$ as follows. For every $j\in F_p^{M'}$, it sends $M_j$'s bits if $j\in H_p^{M'}$ and \metoo\ otherwise.
    In stage 3 of phase $p$, $M$ waits for $k-f$ $\phase \ p \ \stage \ 2$ responses. Then, for every $j\in F_p^M$, if it received only \metoo\ messages, it reassigns $M_j$'s bits evenly between \peers $1, \dots, k$. Otherwise, it updates $\res$ in the appropriate indices. Finally, it continues to stage 1 of phase $p+1$.
    Upon receiving a $\phase \ p \ \stage \ i$ response, $M$ updates $\res$ in the appropriate index and updates $H_p$ for every bit value in the message. See Algorithm \ref{alg:f_crash_async_down} for the formal code.

    Before diving into the analysis, we overview the following intuitive flow of the algorithm's execution. At the beginning of phase 1, the assignment function $\sigma_1$ is the same for every \peer. Every \peer is assigned $n/k$ bits, which it queries and sends to every other \peer. Every \peer $M$ hears from at least $k-f$ \peers, meaning that it has at most $f\cdot n/k$ unknown bits after phase 1. In the following phases, every \peer $M$ reassigns its unknown bits uniformly among all the \peers, such that the bits assigned to every \peer $M'$ are either known to it from a previous phase or $M'$ is about to query them in the current phase (i.e., $M'$ assigned itself the same bits). Hence, after every phase, the number of unknown bits diminishes by a factor of $f/k$. After sufficiently many phases, the number of unknown bits will be small enough to be directly queried by every \peer.

\begin{algorithm}
    \caption{Async Download version 2 for \peer $M$} \label{alg:f_crash_async_down}
    \begin{algorithmic}[1]

    \State \textbf{Local variables} 

    \State $phase(M)$, initially 0 \Comment{This is the present phase of $M$}

    \State $stage(M)$, initially 1 \Comment{This is the present stage of $M$}

    \State $H_p$, $p \in \mathbb{N}$, initially $\emptyset$      \Comment{The set of \peers $M$ heard from during phase $p$}
    \State $\sigma_p(i)$, $p \in \mathbb{N}$, initially $\sigma_p(i) \gets M_{1+\lceil i / \frac{n}{k}\rceil}$  \Comment{The assignment function of $M$ in phase $p$}
    \State $\res[i]$, $i\in \{1,\dots, n\}$, initially $\bot$ \Comment{Output array}

    \State 

    \Upon{entering stage 1 of phase $p$}
        \State \Comment{--------- stage 1 (start) ---------------------------------}
        \State Query all unknown assigned bits, $\{i \mid \res[i]=\bot, \sigma_p(i) = M\}$.
        \State Send a $\phase \ p \ \stage \ 1$ request containing $\{i \mid \res[i] = \bot, \sigma_p(i) = M'\}$ to every \peer
        \State Set $stage(M)\gets 2$

        \State \Comment{--------- stage 2 (start) ---------------------------------}
        \State Wait until $|H_p| \geq k-f$ \label{line:wait_2} 
        \State  Send a $\phase\ p\ \stage\ 2$ request containing $F_p=\{1,\dots, k\}\setminus H_p$ to every \peer $M'$ 
        \State Set $stage(M)\gets 3$
        \State \Comment{--------- stage 3 (start) ---------------------------------}
        \State Wait until received at least $k-f$ $\phase$ $p$ $\stage$ $2$ responses.\label{line:wait_3} 
        \For{$j \in F_p$} 
            \If{Received only \metoo\ responses for $j$}
                \State Let $i_0, \dots, i_{n'-1}$ be the indices such that $\sigma(i_l) = M_j$, $0\leq l\leq n'-1 $.
                \State Set $\sigma_{p+1}(i_l) \gets M_{1+\lceil l / \frac{n'}{k} \rceil}$, $0\leq l\leq n'-1 $.\label{line:reassign}
                \Comment{ Reassign $M_j$'s bits to all $\{1, \dots, k \}$} 
            \Else
                \For{$i \in \{i \mid \sigma(i) = M_j\}$}
                    \State $\res[i] \gets b_i$
                \EndFor
                \Comment{Update $\res$ in the appropriate indices}
            \EndIf
        \EndFor
        \State Set $phase(M)\gets phase(M)+1$, $stage(M)\gets 1$ 
        
    \EndUpon
    
    \State
    
    \Upon{seeing a $\phase \ p \ \stage \ 1$ request for bit Set $B$}
        \State Store the request until $phase(M)=p$ and $stage(M)\geq 2$ or $phase(M) >p$ 

        \State Send back a $\phase \ p \ \stage \ 1$ response containing $\{\langle i, \res[i] \rangle \mid i \in B\}$
        
    \EndUpon

    \State

    \Upon{seeing a $\phase \ p \ \stage \ 2$ request containing $F$ from $M'$}{}
        \State  Store the request until $phase(M)=p$ and $stage(M)\geq 3$ or $phase(M) >p$ 
        \For{$j \in F_p$}
            \State  Send back a $\phase\ p \ \stage \ 2$ response containing \metoo\ if $j\notin H_p$ and $M_j$'s assigned bits otherwise.
        \EndFor  
    \EndUpon

    \State
    \Upon{receiving a $\phase \ p$ $\stage \ i$ response}{}
        \State For every bit value in the message $\langle i, b_i \rangle$, set $\res[i] \gets b_i$
        \State Update $H_p \gets \{j \mid \res[i]\neq \bot \ \ \forall i :\sigma_p(i) = j\}$
    \EndUpon
    \State 
    



    \Upon{$phase(M) = \log_{k/f}(n)$ or $H_p=\{1, \dots, k\}$}{}
    \label{line:termination_condition}
    \State Query all unknown bits
        \State Send $\res$ to every other \peer and Terminate
    \EndUpon
    
\end{algorithmic}
\end{algorithm}

We start the analysis by showing some properties on the relations between local variables.
\begin{observation} \label{obs:no_missing_mahines_no_missing_bits}
    For every \nonfaulty \peer $M$, if $H_p^M=\{1, \dots, k\}$ for some phase $p\geq 0$ then $\res^M=\invec$
\end{observation}
\begin{proof}
    Let $p\geq 0$ be such that $H_p=\{1, \dots, k\}$,
    and consider $1\leq i\leq n$. There exists some $1\leq j\leq k$ such that $\sigma_p(i) = M_j$. Since $j \in H_p^M$, $M$ has heard from $M_j$, so $\res^M[i] \neq \bot$, and overall $\res^M=\invec$.
\end{proof}

Denote by $\sigma^M_p$ the local value of $\sigma_p$ for \peer $M$ at the beginning of phase $p$.
Denote by $\res^M_p[i]$ the local value of $\res[i]$ for \peer $M$ after stage 1 of phase $p$.

\begin{claim} \label{claim:matching_assignment}
For every phase $p$, two \nonfaulty \peers $M, M'$, and bit $i$, one of the following holds.
\begin{enumerate}[$(1_p)$]
    \item $\sigma^M_p(i) =\sigma^{M'}_p(i)$, i.e., both $M$ and $M'$ assign the task of querying $i$ to the same \peer, or
    \item $\res^M_p[i] \neq \bot$ or $\res^{M'}_p[i] \neq \bot$. 
\end{enumerate}
\end{claim}
\begin{proof}
By induction on $p$.
For the basis, $p=0$, the claim is trivially true because of the initialization values (specifically, 
property ($1_0$) holds).

For $p\geq 1$. By the induction hypothesis, either $(1_{p-1})$ or $(2_{p-1})$ holds. 
Suppose first that $(2_{p-1})$ holds, i.e., $\res^M_{p-1}[i] \neq \bot$ or $\res^{M'}_{p-1}[i] \neq \bot$. Without loss of generality, assume that $\res^M_{p-1}[i] \neq \bot$. Then, since values are never overwritten, $\res^M_p[i] \neq \bot$, so $(2_p)$ holds as well.

Now suppose that $(1_{p-1})$ holds, i.e., $\sigma^M_{p-1}(i) =\sigma^{M'}_{p-1}(i)$. Let $j$ be an index such that $\sigma^M_{p-1}(i)=M_j$. If both $M$ and $M'$ didn't hear from $M_j$ during phase $p-1$, then both \peers will assign the same \peer to $i$ in stage 3 of phase $p-1$ (see Line \ref{line:reassign}), so $(1_p)$ holds. If one of the \peers heard from $M_j$, w.l.o.g assume $M$ did, then $\res^M_p[i] \neq \bot$. Hence, $(2_p)$ holds.
\end{proof}

Claim \ref{claim:matching_assignment} yields the following corollary.

\begin{corollary} \label{cor:correct_response}
    Every phase $p$\ stage $1$ request received by a \nonfaulty \peer is answered with the correct bit values.
\end{corollary}

Next, we show that the algorithm never deadlocks, i.e., whenever a \nonfaulty \peer waits in stages 2 and 3 (see Lines \ref{line:wait_2} and \ref{line:wait_3}), it will eventually continue.

\begin{claim} \label{claim:one_nonfaulty_termination}
    If one \nonfaulty \peer has terminated, then every \nonfaulty \peer will eventually terminate.
\end{claim}
\begin{proof}
    Let $M$ be a \nonfaulty \peer that has terminated. Prior to terminating, $M$ queried all the remaining unknown bits and sent all of the bits to every other \peer. Since $M$ is \nonfaulty, every other \nonfaulty \peer $M'$ will eventually receive the message sent by $M$ and will set $H_p^{M'} = \{1, \dots, k\}$, resulting in $\res^{M'}=\invec$ by Observation \ref{obs:no_missing_mahines_no_missing_bits}. Subsequently, $M'$ will terminate as well.
\end{proof}
\begin{claim} \label{claim:finite_wait}
    While no \nonfaulty \peer has terminated, a \nonfaulty \peer will not wait infinitely for $k-f$ responses.
\end{claim}
\begin{proof}
    First, note that at least $k-f$ \peers are \nonfaulty and by Corollary \ref{cor:correct_response} will eventually respond if they see a request before they terminate. Thus, the only case in which a \peer will not get $k-f$ responses is if at least one \nonfaulty \peer has terminated. 
\end{proof}

The combination of Claims \ref{claim:finite_wait} and \ref{claim:one_nonfaulty_termination} implies that eventually, every \nonfaulty \peer satisfies the termination condition (see Line \ref{line:termination_condition}) and subsequently terminates correctly (since it queries all unknown bits beforehand). That is because by Claim \ref{claim:finite_wait} some \nonfaulty \peer $M$ will get to phase $\log_{k/f}(n)$, or set $H_p^M=\{1, \dots, k\}$ prior to that, and terminate, which will lead to the termination of every \nonfaulty \peer by Claim \ref{claim:one_nonfaulty_termination}.


\begin{claim} \label{claim:unkown_bits_per_phase}
    At the beginning of phase $p\geq 0$, every \nonfaulty \peer has at most $n\cdot \left(\frac{f}{k}\right)^p$ unknown bits.
\end{claim}
\begin{proof}
    By induction on $p$. Consider \nonfaulty \peer $M$.
    For the base step $p=0$ the claim holds trivially by the initialization values.
    
    Now consider $p\geq 1$.
    By the induction hypothesis on $p-1$, $M$ has at most $\hat{n} = n\cdot \left(\frac{f}{k}\right)^{p-1}$ unknown bits at the start of phase $p-1$.
    Since unknown bits are assigned evenly in stage $3$ (see Line \ref{line:reassign}), each \peer is assigned $\hat{n}/k$ unknown bits (to be queried during phase $p-1$).
    During stage 2 of phase $p-1$, $M$ waits until $|H_{p-1}^M| \geq k-f$, meaning that $M$ did not receive the assigned bits from at most $f$ \peers.
    Hence, there are at most $\hat{n}/k \cdot f = n\cdot \left(\frac{f}{k}\right)^p$ unknown bits after stage 2 of phase $p$. The claim follows.
\end{proof}

By Claim \ref{claim:unkown_bits_per_phase} and since unknown bits are distributed evenly among $\{0, \dots, k-1\}$, every \nonfaulty \peer queries at most $\frac{n}{k}\cdot \left(\frac{f}{k}\right)^{p}$ in phase $0\leq p\leq \log_{k/f}(n)$ and at most $\frac{n}{k}\cdot \left(\frac{f}{k}\right)^{\log_{k/f}(n)} = \frac{1}{k}$ additional bits when terminating (By Observation \ref{obs:no_missing_mahines_no_missing_bits}). Hence, the worst case query complexity (per \peer) is bounded by 
$$\Query ~\leq~ \frac{1}{k} + \sum_{p=1}^{\log_{k/f}(n)} \frac{n}{k}\cdot \left(\frac{f}{k}\right)^{p} ~=~ O\left(\frac{n}{\goodfrac k}\right).$$

We next turn to time analysis.
Let $M$ be a \peer. For every phase $p$, after $n/k \cdot (\frac{f}{k})^p$ time, every $\phase\ p\ \stage\ 1$ response by a \nonfaulty {\peer} is heard by $M$ (even slow ones),
and stage 2 starts. After that, it takes at most $n \cdot (\frac{f}{k})^{p+1}$ time units for every $\phase\ p\ \stage\ 2$ response to be heard by $M$, allowing it to move to stage 3. Hence, it takes at most $n/k \cdot (\frac{f}{k})^p + n \cdot (\frac{f}{k})^{p+1}$ time for phase $p$ to finish once $M$ started it. Finally, upon termination, $M$ sends $\res_M$ which takes $n$ time. Overall the time complexity is 

\[
\Time ~\le~ n + \sum_{p=0}^{\log_{k/f}(n)} \left(\frac{n}{k} \cdot \left(\frac{f}{k}\right)^p + n \cdot \left(\frac{f}{k}\right)^{p+1}\right) ~=~ n + O\left(\frac{n\cdot f}{\goodfrac k}\right) ~=~ O\left(\frac{\byzfrac}{\goodfrac}\cdot n\right).
\]

This results in the following.
\begin{lemma}
Algorithm \ref{alg:f_crash_async_down} solves Download in the asynchronous setting with at most $f$ crash faults after $\log_{k/f}(n)$ phases with $\Query = O(\frac{n}{\goodfrac k})$ and $\Time = n + O(\frac{\byzfrac}{\goodfrac}\cdot n)$.
\end{lemma}

The following explains how we get better time complexity 
Note that this requires a slight modification of the code.
To make the time complexity analysis more precise, we identify a necessary condition for a \peer to send bits in a $\phase\ p \ \stage\ 2$ response (rather than 
\metoo). We observe that after 1 time unit, every message is delivered (even by slow \peers). Hence, after at most $\frac{n}{k}\cdot \left(\frac{f}{k}\right)^{p}$ time units, every $\phase\ p\ \stage\ 1$ response that was sent is delivered (including slow ones).
Therefore, while waiting for $k-f$ $\phase\ p\ \stage\ 2$ responses, it might be the case that a slow $\phase\ p\ \stage\ 1$ response arrives from \peer $M'$ eliminating the need for $\phase\ p\ \stage\ 2$ responses regarding \peer $M'$. 
The modification needed for this argument to work is that if $\phase\ p\ \stage\ 2$ responses regarding \peer $M'$ are no longer necessary because of its $\phase\ p\ \stage\ 1$ response arriving, $M$ is not blocked from continuing. Note that it is easy to see that this modification doesn't effect the correctness of the protocol.
Hence, the only time when a \peer must wait for a long $\phase\ p\ \stage\ 2$ response is when the corresponding \peer for which a $\phase\ p\ \stage\ 1$ response was not received has crashed (and therefore its $\phase\ p\ \stage\ 1$ response will never arrive). Also note that once a \peer crashed in phase $p$, it will not be heard from by any \peer in following phases resulting in 
\metoo\
responses. Therefore every \peer waits for long $\phase\ p\ \stage\ 2$ responses at most $f$ times.
This results in the complexity being
\[\Time ~\leq~ n + \beta n + \sum_{p=0}^{\log_{k/f}(n)} \frac{n}{k} \cdot \left(\frac{f}{k}\right)^p ~=~ O\left(n + \frac{n}{\goodfrac k}\right) ~=~ O(n) .
\]

\begin{theorem} \label{thm:async_multi_crash}
There is a deterministic algorithm for solving Download in the asynchronous setting with at most $f$ crash faults 
with $\Query = O(\frac{n}{\goodfrac k})$ and $\Time = O(n)$.
\end{theorem}

\section{Conclusions and Future Work}\label{sec: conclusion}
In this work, we studied the Data Retervial model, introduced in \cite{ABMPRT24}, and improved the previous results on the Download problem by achieving both optimal query complexity and optimal resiliency of any fraction $\byzfrac < 1$ of byzantine \peers. In addition, we presented several new results, including a protocol with $O(\log n)$ time complexity and near optimal expected query complexity in a model with a dynamic adaptive adversary, and, in the Broadcast model, near optimal worst case time, query, and message complexity. We also established a lower bound for single-round protocols, demonstrating that it is necessary to query every bit to solve the download problem within a single round. Further, we solved the download problem with optimal query complexity in synchronous and asynchronous networks in the presence of crash faults.

Our work raised several intriguing questions for Byzantine fault model, especially regarding the lower bounds on query and time complexity and the query-time tradeoff. 
A natural extension of this research would be to explore the optimal message complexity for the point-to-point model (i.e., without Broadcast). In an asynchronous network, for deterministic algorithms that handle crash faults, a key question is whether both time and query optimal complexities can be achieved simultaneously. Addressing these open questions could lead to a deeper understanding and refinement of efficient protocols in this domain.

Our work is useful in the context of blockchain oracles, specifically, as a sub-routine for data extraction from multiple data sources, where some of those data sources have some (possibly probabilistic) guarantee of being nonfaulty (or honest). A relevant question in this context involves handling data that changes over time, such as stock prices or exchange rates. We leave this question open for future exploration, as formalizing these temporal changes may require adjustments to the problem’s existing constraints.

\bibliographystyle{plain}
\bibliography{references}

\end{document}